\documentclass[a4paper,onecolumn,11pt,accepted=2022-10-17]{quantumarticle}
\pdfoutput=1
\usepackage[utf8]{inputenc}
\usepackage[english]{babel}
\usepackage[T1]{fontenc}
\usepackage{amsthm}
\usepackage{amsfonts}
\usepackage{amssymb}
\usepackage{amscd}
\usepackage{amsmath}    % need for subequations
\usepackage{enumerate}
\usepackage{epsfig}
\usepackage{subfigure}
\usepackage{subfloat}
\usepackage{xcolor}			
\usepackage{physics}
\usepackage{lipsum}
\usepackage[most]{tcolorbox}
\usepackage[]{qcircuit}
\usepackage{tabularx}
\usepackage{appendix}
\usepackage{ulem}
\usepackage{hyperref}

\usepackage[numbers,sort&compress]{natbib}

\newtcolorbox[auto counter]{game}[2]{
	label=#1,
	title=Box \thetcbcounter #2,
	colback=blue!10,
	colframe=blue!50!black
}

\newtheorem{theorem}{Theorem}
\newtheorem{lemma}{Lemma}
\newtheorem{corollary}{Corollary}

\newtheorem{definition}{Definition}

\begin{document}
\title{Classically Replaceable Operations}
%\date{}
\author{Guoding Liu}
\orcid{0000-0003-3288-1394}
\author{Xingjian Zhang}
\orcid{0000-0003-0677-6996}
%\email{zhang-xj18@mails.tsinghua.edu.cn}
\author{Xiongfeng Ma}
\orcid{0000-0002-9441-4006}
\email{xma@tsinghua.edu.cn}
\affiliation{Center for Quantum Information, Institute for Interdisciplinary Information Sciences, Tsinghua University, Beijing, 100084 China}

\begin{abstract}
Quantum information science provides powerful technologies beyond the scope of classical physics. In practice, accurate control of quantum operations is a challenging task with current quantum devices. The implementation of high fidelity and multi-qubit quantum operations consumes massive resources and requires complicated hardware design to fight against noise. An approach to alleviating this problem is to replace quantum operations with classical processing. Despite the common practice of this approach, rigorous criteria to determine whether a given quantum operation is replaceable classically are still missing. In this work, we define the classically replaceable operations in four general scenarios. In each scenario, we provide their necessary and sufficient criteria and point out the corresponding classical processing. For a practically favorable case of unitary classically replaceable operations, we show that the replaced classical processing is deterministic. Beyond that, we regard the irreplaceability of quantum operations by classical processing as a quantum resource and relate it to the performance of a channel in a non-local game, as manifested in a robustness measure.
\end{abstract}

\maketitle

\section{Introduction}
Quantum technology realizes information processing tasks that cannot be achieved by classical means. For instance, quantum key distribution (QKD) brings information-theoretic security between two remote parties~\cite{bennett1984quantum,ekert1991Quantum}. In the field of computation, quantum computers have the potential to bring an exponential speedup in solving certain problems. The recent realizations of controllable quantum systems with dozens of qubits exhibit the ``quantum advantage'' over their classical counterparts~\cite{Arute2019,HanSenZhong2020Bosonic,YulinWu2021Superconducting}. On the other hand, quantum information processing is costly. Due to the unavoidable noise in quantum devices, current realizations of high-fidelity multi-qubit quantum operations are extremely challenging. On the contrary, classical processing is in general easier to implement and has higher accuracy.

The tractability of classical processing motivates us to replace some quantum operations in a task with classical ones. In QKD, after the quantum stage of state preparation and measurement, the users would apply certain quantum operations to distill a secret key~\cite{Lo2050Unconditional}. In general, these operations contain a mass of multi-qubit gates, like controlled-NOT and Tofolli gates, and hence are very challenging to realize. Thanks to Shor and Preskill's seminal work to reduce quantum operations to classical post-processing~\cite{Shor2000Proof}, QKD becomes the earliest practical application in quantum information science~\cite{2020QKDreview}.

The idea of replacing is also conducive to quantum computing. Within the decoherence time of a quantum system, the number of implementable operations is bounded~\cite{nielsen2002quantum}, hindering the ability of quantum computing to solve large-scale problems. If we replace some quantum operations in a quantum circuit with classical processing, we can reduce the number of quantum gates. For instance, when realizing a quantum circuit, we shall omit any operation commuting with the observable to be measured, as these operations do not change the measurement results.

These instances show that certain quantum operations under specific scenarios can be replaced by classical processing. In practice, classical replacement is flexible and can be embedded into many quantum information processing protocols. In particular, it is appealing to combine the idea with the recent quantum-classical-hybrid algorithms and data processing techniques, such as the variational quantum algorithms (VQA)~\cite{Cerezo2021Variational} and shadow tomography~\cite{huang2020shadow}, where we may further reduce the difficulty in the quantum part of the task.

Here, we are interested to see what kind of quantum operations are classically replaceable, namely, classically replaceable operations (CROs). In the study of some resource theories, such as quantum coherence of states~\cite{Eric2016DIO,LiuZiWen2017ResourceDestroying,gour2017quantum} and entanglement of quantum channels~\cite{Horodecki2003EBChannel,rosset2018resource}, the issue has been implicitly involved. For example, maximally incoherent operations (MIO)~\cite{Baumgratz2014Coherence,ABERG2004Subspace,Aberg2006BlockCoherence} cannot generate coherence from incoherent states, which means that on the set of incoherent states they act like classical processing. This implies that MIO is classically replaceable to some extent. In the entanglement theory, entanglement breaking (EB) channels destroy the entanglement of the input states. This also shows a clue that they are classically replaceable in some scenarios.

In the literature, these operations are mainly studied under the scenario of coherence manipulation or entanglement manipulation. At the moment, a systematic study of CROs is still missing. There are neither rigorous definitions nor efficient criteria for a CRO. On the other hand, there are quantum operations that cannot be replaced classically. Then, this ``irreplaceability'' implies quantum features intractable to classical means, which might exhibit quantum advantages in some information processing tasks. A mathematical characterization of the irreplaceability would reveal the boundary between quantum and classical technologies.

In this work, we provide rigorous definitions of CROs in four different scenarios, depending on whether the input and output of the operation are classical or quantum. It is worth noting that, whether a quantum operation can be precisely replaced is what we are concerned with for classical replacement. In the case where the input and output are both classical, classical replacement is essentially classical simulation and any quantum operation is a CRO. We shall focus on the remaining three cases where at least either the input or the output is quantum. We characterize the three CRO sets by showing their necessary and sufficient criteria. Furthermore, we establish a resource-theoretic framework to quantify the irreplaceability of an operation. Interestingly, we can prove the existence of a non-local game to show the quantum advantage of irreplaceable operations.

The paper is organized as follows. In Section~\ref{sc:CROdef}, we clarify the definitions of ``classical processing'' and ``classical replaceability'' in this work. In Section~\ref{sc:CROmath}, we show the main result -- providing the mathematical characterization of CROs. Furthermore, we deal with two extensions of CRO and discuss the application of CRO. In Section~\ref{sc:Resource}, we establish a channel resource theory for CRO and quantify the irreplaceability of a channel in a non-local game. Finally, we conclude in Section~\ref{sc:Conclusion} with further discussions.

\section{Classical replaceability}\label{sc:CROdef}
First, we introduce the notations. A quantum system is represented as a state in a Hilbert space $\mathcal{H}$. For simplicity, we assume the dimension of $\mathcal{H}$ is finite. Denote $d = \dim \mathcal{H}$, we could find $d$ orthonormal vectors $\{\ket{e_i}, i\in [d]\}$ as a basis of $\mathcal{H}$, where $[d] \equiv \{0,1,\cdots,d-1\}$. The basis of $\mathcal{H}$ is not unique and we usually define one basis $\{\ket{i}, i\in [d]\}$ as the computational basis of $\mathcal{H}$. If there are $n$ quantum systems $\mathcal{H}_0,\mathcal{H}_1,\cdots, \mathcal{H}_{n-1}$, the computational basis of $\bigotimes_{k=0}^{n-1} \mathcal{H}_k$ is defined as the tensor product of computational basis on each subsystem. Denote $\mathcal{D}(\mathcal{H})$ as the set of all density operators on $\mathcal{H}$. A quantum operation, $O$, is defined as a completely positive and trace preserving (CPTP) map acting on $\mathcal{D}(\mathcal{H})$. That is, $\forall \rho\in\mathcal{D}(\mathcal{H}'\otimes \mathcal{H})$, and $\forall\sigma \in \mathcal{D}(\mathcal{H})$,
\begin{equation}
\begin{split}
I'\otimes O(\rho)&\geq 0,\\
\tr(O(\sigma)) &= \tr(\sigma),
\end{split}
\end{equation}
where $\mathcal{H}'$ is an ancillary quantum system with an arbitrary dimension and $I'$ is the identity operation acting on $\mathcal{D}(\mathcal{H}')$.
In the following, we denote CPTP as the set of all quantum operations. The multiplication or composition on CPTP, is denoted as $\circ$.

A classical system can be viewed as a random variable. If a random variable has $d$ values, we call it as a dit with dimension $d$. Given a dit taking the value $i$ with probability $p_i$, $i\in[d]$, $\sum_i p_i = 1$, the state preparation under basis $\{\ket{e_i}, i\in [d]\}$ is transforming the dit into the quantum state $\sigma = \sum_{i=0}^{d-1}p_i \ketbra{e_i}$.
Reversely, given a quantum state $\rho$, we can perform a measurement under basis $\{\ket{e_i}, i\in [d]\}$ and obtain a random variable taking the value $i$ with probability $\tr(\rho \ketbra{e_i})$.

Now, let us clarify the meaning of ``classical processing''. Classical computers can be viewed as Turing machines or circuit models~\cite{turing1936computable, savage1972computational}. Any function of the form $f:[d]\rightarrow [d^{\prime}]$ can be realized with a certain number of fixed logical gates in classical circuits~\cite{nielsen2002quantum}, where $d, d'\in \mathbb{Z}_{+}$ are positive integers. If we introduce probabilistic Turing machines~\cite{Santos1969Probabilistic, Gill1977Probabilistic}, we have the ability to compute random functions. In principle, we can use a classical computer to realize any probabilistic function subject to a pre-determined probability distribution. In this work, we do not consider the computability of probabilistic Turing machine. That is, we do not consider the computational complexity to realize a random function.

Mathematically, a random function can be described by a stochastic matrix. Given a dit $s$ taking $i$ with probability $p_i$, a probabilistic map $O_c:[d]\rightarrow [d^{\prime}]$ transfers it to a new random variable $s^{\prime} = O_c(s)$. A specific result $j$ occurs with probability
\begin{equation}
\Pr(s^{\prime} = j) = \sum_i T_{ji}p_i,
\end{equation}
where the probabilistic map is described by the stochastic matrix $T$, satisfying $T_{ji}\geq 0$ and $\sum_{j}T_{ji} = 1, \forall i$. In general, $d$ and $d^{\prime}$ can be different.

Now we aim to find out all the quantum operations that can be replaced by classical processing. We distinguish CROs in four cases. To be specific, we restrict the input and output of the operation to be classical or quantum, and then find out the CROs. They are classical-quantum CROs (cqCRO), quantum-quantum CROs (qqCRO), quantum-classical CROs (qcCRO), and classical-classical CROs (ccCRO). Here, classical input means that we prepare a state based on a classical random variable and set this state as the input of the quantum operation. Classical output means that we further implement a measurement after the action of a quantum operation. The measurement result is the classical output and we do not concern about the remained state after the measurement. The replacement means that the output does not change for any input, even the input is unknown.

\begin{definition}[Classically replaceable operation]\label{def:CRO}
There are four cases that a quantum operation is replaceable with classical processing, cqCRO, qqCRO, qcCRO, and ccCRO defined as follows, as shown in Figure~\ref{fig:CROdef}. Here, we define a fixed computational basis for state preparation and measurement.

\begin{enumerate}
\item
Consider operations right after a state preparation. A cqCRO can be realized by first preprocessing the input classically and then preparing the output state.

\item
Consider operations with quantum input and quantum output. A qqCRO can be realized by first measuring the input, processing the outcomes classically, and then preparing the output state.

\item
Consider operations right before a measurement. A qcCRO can be realized by first measuring the input state and then processing the outcomes classically.

\item
Consider operations between a state preparation and a measurement. A ccCRO can be realized by processing the input classically.
\end{enumerate}
\end{definition}

\begin{figure}[!htbp]
	\centering \resizebox{10cm}{!}{\includegraphics{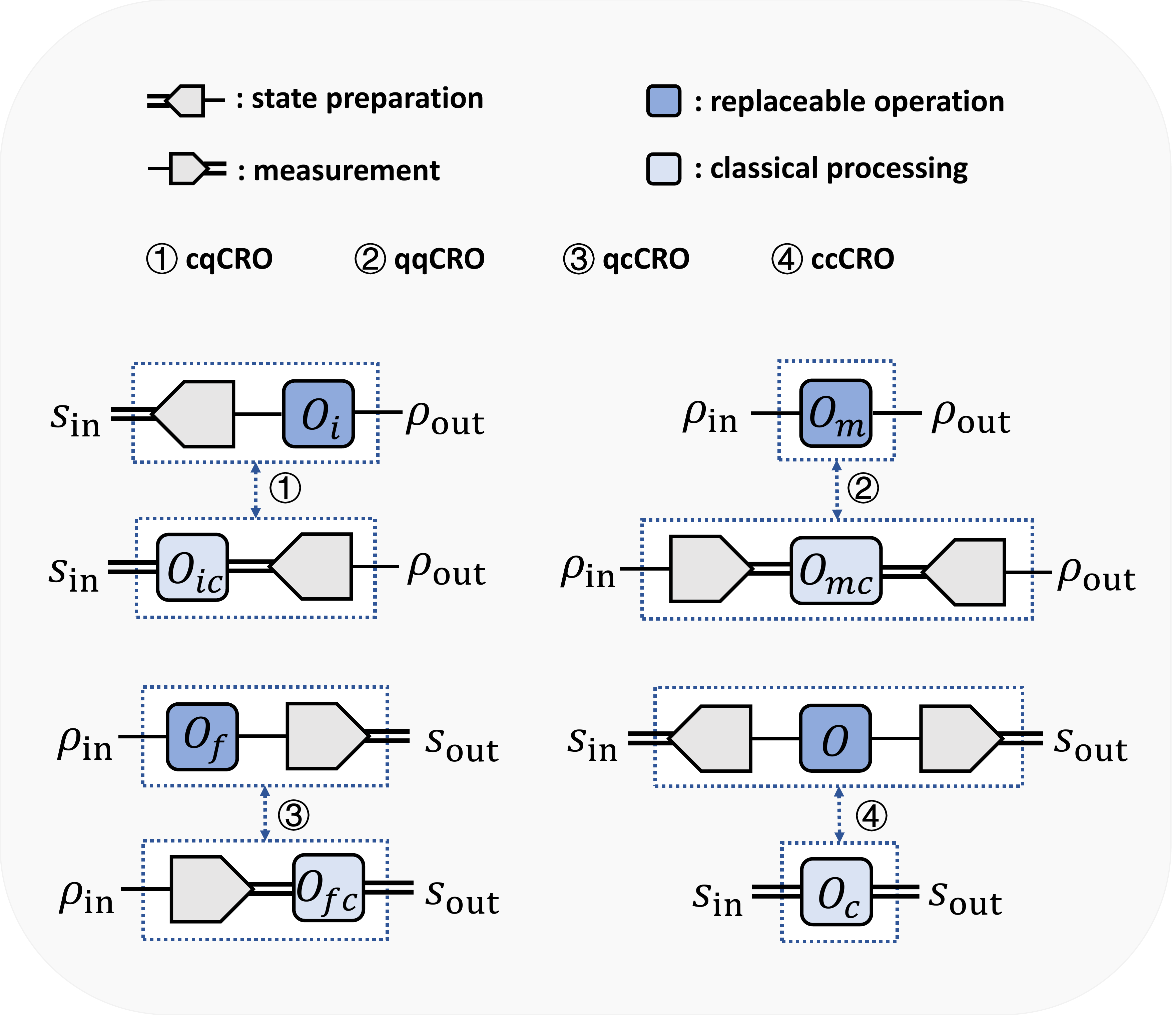}}
	\caption{Four different kinds of CROs. The CROs are the quantum operations in dark blue blocks excluding the state preparation and measurement. Here, all state preparation operations and measurements are performed on the computational basis. \textcircled{1} Given any classical input $s_{\mathrm{in}}$, the output quantum state $\rho_{\mathrm{out}}$ of a cqCRO, $O_i$, can be obtained by first preprocessing $s_{\mathrm{in}}$ with classical processing $O_{ic}$ and then preparing the output state. \textcircled{2} Given any input quantum state $\rho_{\mathrm{in}}$, the output quantum state $\rho_{\mathrm{out}}$ of a qqCRO, $O_m$, can be obtained by first measuring the state, applying classical processing $O_{mc}$ to the outcomes, and preparing the state. \textcircled{3} Given any input state $\rho_{\mathrm{in}}$, the measurement result $s_{\mathrm{out}}$ after performing a qcCRO, $O_f$, can be obtained by first measuring $\rho_{\mathrm{in}}$ and then processing the outcomes classically with $O_{fc}$. \textcircled{4} Given any classical input $s_{\mathrm{in}}$, the classical output $s_{\mathrm{out}}$ of a ccCRO, $O$, can be obtained by processing $s_{\mathrm{in}}$ classically with $O_c$.}
	\label{fig:CROdef}
\end{figure}

It is worth mentioning that in Definition~\ref{def:CRO}, the CRO is the quantum operation \textit{beside} the state preparation and measurement where the state preparation and measurement are \textit{excluded}. We require the most strict condition for qqCRO as we do not put any restriction on the input and output. Any qqCRO is a cqCRO and a qcCRO. Interestingly, in Section~\ref{sc:CROmath}, we show that the intersection of cqCRO and qcCRO is not qqCRO. Furthermore, the case of ccCRO can be viewed as the classical simulation. As all quantum operations can be simulated classically, ccCRO contains all quantum operations so that cqCRO, qqCRO, and qcCRO are all subsets of ccCRO. Note that in this work we might use CRO to represent a replaceable operation or the set of replaceable operations depending on the context.

In Definition~\ref{def:CRO}, a CRO only guarantees that the subsystem it acts on is unchanged before and after classical replacement. The state of the rest part of the system might change. Assume that a CRO on system $A$, $O_A$, originates from a unitary, $U_{AB}$, on a large system, $AB$. In general, replacing $O_A$ will destroy $U_{AB}$ and change the correlation between $A$ and $B$. On the other hand, if we consider a third system beyond $AB$, $C$, replacing $O_A$ will not change the correlation between $A$ and $C$. Strictly speaking, the evolution on system $A$ originating from $U_{AB}$ is not necessarily a CPTP map. Here, we only consider the CPTP map case. We provide more discussions in Appendix~\ref{append:PartReplace}.

Now, let us consider unitary CROs, namely, classically replaceable unitary operations (CRU). This is a special case when $B$ is a trivial system. In Section~\ref{sc:CROmath}, we prove that cqCRU, qqCRU, and qcCRU can be replaced by deterministic classical processing while this conclusion is not true for ccCRU. A counter-example is that a Hadamard gate between the state preparation and measurement on a qubit can only be replaced with probabilistic classical processing. Figure~\ref{fig:QCircuit} shows examples of CRUs in a quantum circuit. The state of the whole quantum circuit at any time does not change before and after replacement for any circuit input.

\begin{figure}[!htbp]
\centering \resizebox{15cm}{!}{\includegraphics{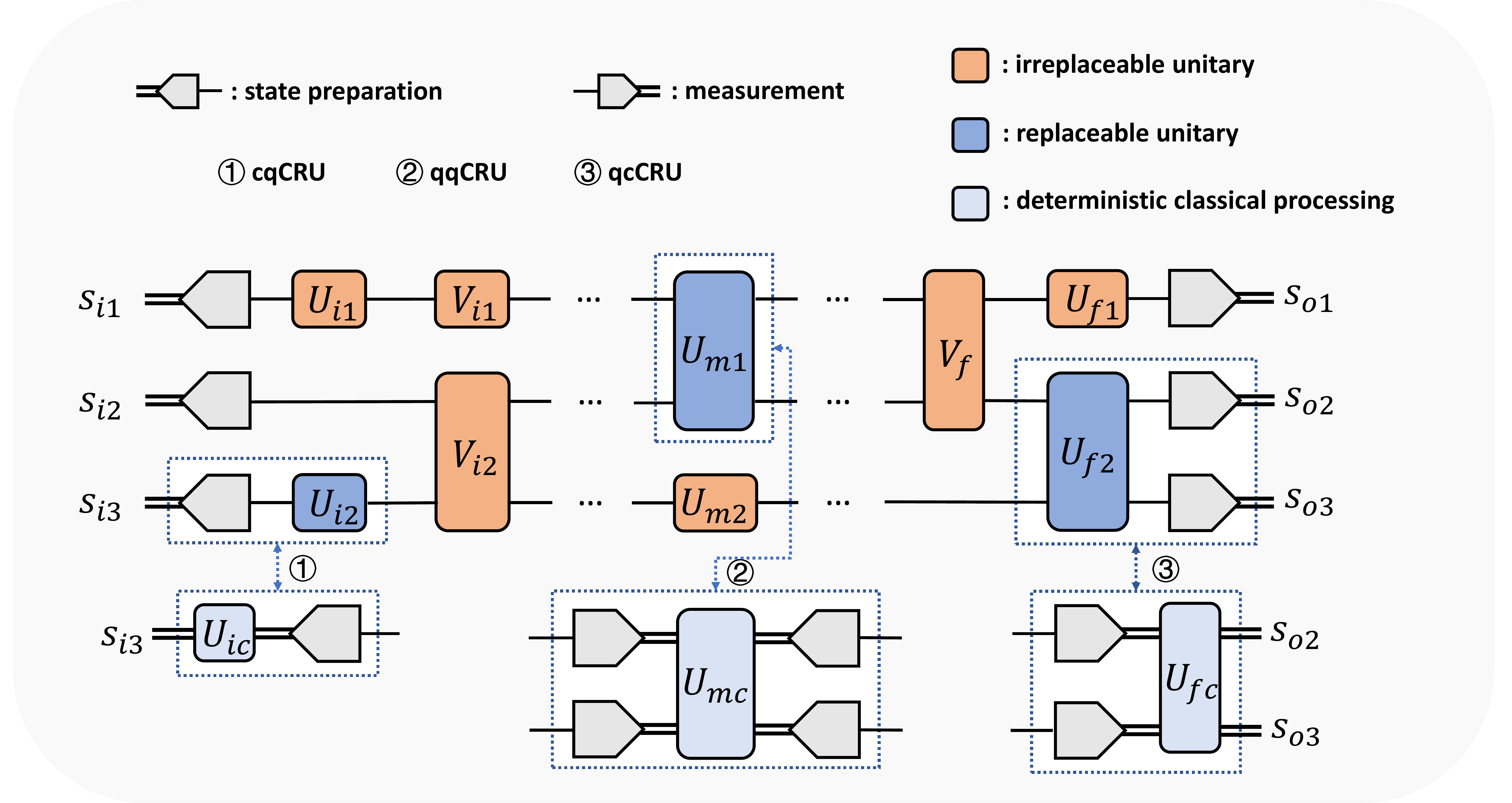}}
\caption{Within a quantum circuit, cqCRU, qqCRU, and qcCRU always lie in the beginning, middle, and ending of the circuit, respectively. All state preparation operations and measurements are performed on the computational basis. For any input statistics $(s_{i1},s_{i2},s_{i3})$, the state of the whole quantum circuit at any time does not change before and after replacing CRU $U_{i2},U_{m1}$, and $U_{f2}$ with $U_{ic},U_{mc}$, and $U_{fc}$, respectively. The statistics of output $(s_{o1},s_{o2},s_{o3})$ also remains the same. Note that the classical replacement can be performed on a part of the qubits. For example, even if $U_{f1}\otimes U_{f2}$ is not a qcCRU, we can view it as $(U_{f1}\otimes \mathbb{I})(\mathbb{I}\otimes U_{f2})$ and replace qcCRU $U_{f2}$.}
\label{fig:QCircuit}
\end{figure}

\section{Mathematical Characterization of CRO}\label{sc:CROmath}
In this section, we first provide the mathematical characterizations for the four kinds of CROs. As any quantum operation can be treated as a ccCRO, we focus on the rest three cases. Then, we discuss the extensions of CRO when the state preparation and measurement are not performed on the computational basis.

\subsection{Equivalent Definitions of CRO}
Here, we provide the mathematical characterization of the first three kinds of CROs in Definition~\ref{def:CRO}. Given a computational basis on $\mathcal{H}$, the dephasing operation is defined as a map from $\mathcal{D}(\mathcal{H})$ to $\mathcal{D}(\mathcal{H})$ that $\forall \rho\in \mathcal{D}(\mathcal{H})$,
\begin{equation}
	\Delta (\rho) =  \sum_{i=0}^{d-1} \ketbra{i}\rho\ketbra{i}.
\end{equation}
Then, we mathematically characterize the three kinds of CROs.
\begin{theorem}\label{thm:CROequiv}
For any $O\in \mathrm{CPTP}$,
\begin{itemize}
	\item $O\in\mathrm{cqCRO}\Leftrightarrow O\circ\Delta = \Delta\circ O\circ \Delta\Leftrightarrow \exists O^{\prime}\in\mathrm{CPTP}, O\circ \Delta = \Delta\circ O^{\prime}\circ \Delta$;
	\item $O\in\mathrm{qqCRO}\Leftrightarrow O = \Delta\circ O\circ \Delta\Leftrightarrow \exists O^{\prime}\in\mathrm{CPTP}, O = \Delta\circ O^{\prime}\circ \Delta$;
	\item $O\in\mathrm{qcCRO}\Leftrightarrow \Delta\circ O = \Delta\circ O\circ \Delta\Leftrightarrow \exists O^{\prime}\in\mathrm{CPTP}, \Delta\circ O = \Delta\circ O^{\prime}\circ \Delta$.
\end{itemize}
\end{theorem}

We sketch the proof here and leave the details in Appendix~\ref{thmproof:CROequiv}. As the proof for three kinds of CROs are similar, we exhibit the idea with respect to qcCRO. Theorem~\ref{thm:CROequiv} contains three equivalent criteria for each kind of CRO. We first prove the equivalence of the second criterion and the third, then the first and the second. The equivalence of the second criterion and the third can be directly verified with the idempotence of the dephasing operation.

For the next proof, we first find the corresponding classical processing for each operation satisfying $\Delta\circ O = \Delta\circ O\circ \Delta$. Then we prove that if any operation does not satisfy that criterion, then it cannot be replaced. The argument for the irreplaceability of an operation comes from that it can output different results for two inputs but through classical processing we can only get the same results. From this we can conclude there exists a quantum state that the output cannot be obtained after replacing the operation with classical processing.

From Theorem~\ref{thm:CROequiv}, we can define the sets for different CROs,
\begin{align}
\label{eq:cqCRO} &\mathrm{cqCRO} = \{O\in \mathrm{CPTP}|O\circ\Delta = \Delta\circ O\circ \Delta\},\\
\label{eq:qqCRO} &\mathrm{qqCRO} = \{O\in \mathrm{CPTP}|O = \Delta\circ O\circ \Delta\},\\
\label{eq:qcCRO} &\mathrm{qcCRO} = \{O\in \mathrm{CPTP}|\Delta\circ O = \Delta\circ O\circ \Delta\}.
\end{align}

In Appendix~\ref{thmproof:CROequiv}, we show that for any qcCRO, $O$, the stochastic matrix $T$ of the corresponding classical processing satisfies $T_{ji} = \tr(\ketbra{j}O(\ketbra{i}))$. If $\exists U$, where $U$ is unitary, such that $\forall\rho\in \mathcal{D}(\mathcal{H})$, $O(\rho) = U\rho U^{\dagger}$. Then Eq.~\eqref{eq:qcCRO} requires that
\begin{equation}
\sum_k \ketbra{k} U\ketbra{i}{j} U^{\dagger} \ketbra{k} = \delta_{ij}\sum_k \ketbra{k} U\ketbra{i} U^{\dagger} \ketbra{k},
\end{equation}
where $\delta_{ij}$ equals 1 if $i = j$ and $0$ otherwise. Then for any $i,j,k$ and $i\neq j$,
\begin{equation}
U_{ki}U_{kj}^* = 0,
\end{equation}
where $U_{ki} = \bra{k}U\ket{i}$. It means for each row of $U$, there is only one element nonzero. As $T_{ji} = \tr(\ketbra{j}U\ketbra{i}U^{\dagger}) = \abs{U_{ji}}^2$, the value of $T_{ji}$ must be 0 or 1. Then the replacing classical processing is deterministic. Similar arguments apply to unitary operations in cqCRO and qqCRO,
\begin{corollary}
cqCRU, qqCRU and qcCRU can be replaced with deterministic classical processing.
\end{corollary}

It is interesting that from Theorem~\ref{thm:CROequiv}, we can obtain a relation between CRO and the resource theory of coherence~\cite{Baumgratz2014Coherence}. In the framework of coherence resource theory, the free states, also named incoherent states, are the states diagonal on the computational basis. We can define two sets of operations: MIO and coherence non-activating operations (CNAO)~\cite{LiuZiWen2017ResourceDestroying}.
The former is the maximal set of operations that cannot generate coherence from incoherent states. The latter is the maximal set of operations that cannot activate coherence from the input state. The details of coherence theory are shown in Appendix~\ref{appendsc:resource}.

In fact, Eq.~\eqref{eq:cqCRO} can be seen as the equivalent definition of MIO~\cite{LiuZiWen2017ResourceDestroying}. Similarly, Eq.~\eqref{eq:qcCRO} is a definition for CNAO. Then, we see the equivalence between cqCRO and MIO, qcCRO and CNAO. The intersection of cqCRO and qcCRO, or the intersection of MIO and CNAO, is a set of dephasing-covariant incoherent operations (DIO)~\cite{Eric2016DIO,LiuZiWen2017ResourceDestroying}, defined as
\begin{equation}\label{eq:DIO}
\mathrm{DIO} = \{O\in \mathrm{CPTP}|\Delta\circ O = O\circ \Delta\}.
\end{equation}

It is worth noting that although MIO and CNAO have been proposed before, their classical replaceability has not been discovered yet. In a sense, the coincidence of the equivalence between cqCRO and MIO, qcCRO and CNAO, respectively, endows MIO and CNAO with a new operational meaning.

From Eq.~\eqref{eq:qqCRO} and Eq.~\eqref{eq:DIO}, we can see that qqCRO is a proper subset of DIO, hence qqCRO is not the intersection of cqCRO and qcCRO. An operation belonging to both cqCRO and qcCRO might not be qqCRO. For example, CNOT gate is a DIO but not a qqCRO. At the same time, qqCRO is a subset of EB channels~\cite{Horodecki2003EBChannel}, since any qqCRO, $O$, can be given by
\begin{equation}
O(\rho) = \sum_{ij} \tr(\rho \ketbra{i}) T_{ji}\ketbra{j},
\end{equation}
where $T_{ji}$ is the stochastic matrix associated with the replaced classical processing. Furthermore, DIO is a proper subset of both cqCRO and qcCRO. For example, $\Delta\circ \text{Had}$ is a cqCRO but not a DIO, where Had represents the super-operator form of the Hadamard gate. A state preparation channel, $O(\rho) = \ketbra{+}$, is a qcCRO but not a DIO. Also, we can find an EB channel, $O(\rho) = \ketbra{0}{+}\rho \ketbra{+}{0} + \ketbra{+}{-}\rho \ketbra{-}{+}$, which is not a CRO. In summary, the relations among three kinds of CROs, MIO, CNAO and EB channels can be visualized in Figure~\ref{fig:CROVenn}.

\begin{figure}[!htbp]
	\centering \resizebox{8cm}{!}{\includegraphics{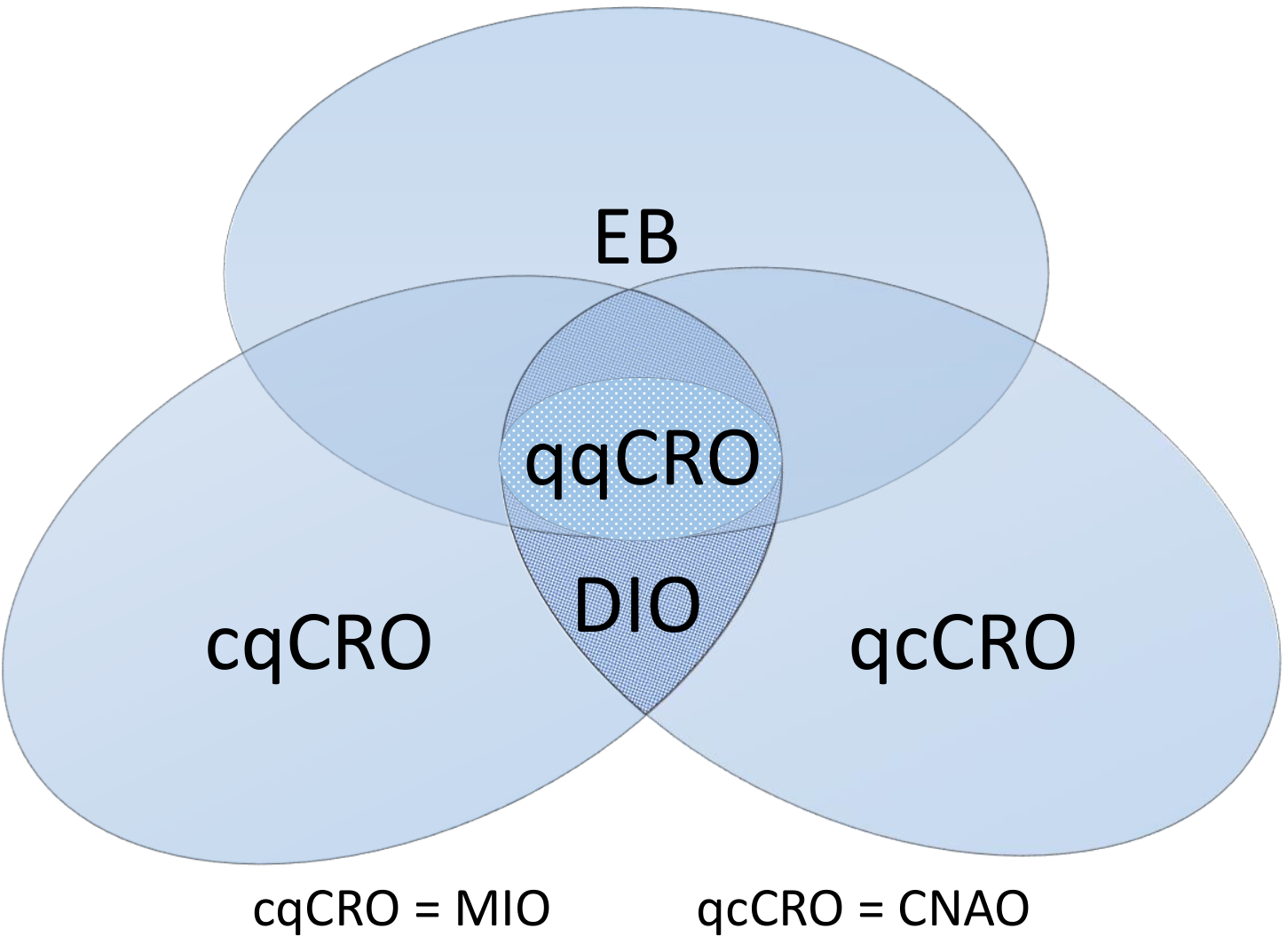}}
	\caption{Relations among three kinds of CROs, MIO, CNAO and EB channels. Here, cqCRO = MIO, qcCRO = CNAO, DIO = $\mathrm{cqCRO} \cap \mathrm{qcCRO}$, $\mathrm{qqCRO} \subset \mathrm{DIO}$, $\mathrm{qqCRO} \subset \mathrm{EB}$.}
	\label{fig:CROVenn}
\end{figure}

\subsection{Extension of CRO to Projective Measurement}\label{ssc:CROpvm}
In previous discussions, the measurements are rank-one, or, von Neumann measurements. Here, we discuss a more general case where the measurement is a projective measurement, described by the projective-valued measure (PVM). A PVM on a $d$-dimensional Hilbert space $\mathcal{H}$ is defined with a set of projectors on the space, $\{E_n| n\in[d']\}$, where $d'\leq d$, $E_nE_m = \delta_{nm}E_n$, and $\sum_n E_n = \mathbb{I}$. Generally, the measurement on a state $\rho\in \mathcal{D}(\mathcal{H})$ can be viewed as a CPTP map $\mathcal{E}$ from $\mathcal{H}$ to $\mathcal{H}\otimes \mathcal{H}'$, where $\mathcal{H}'$ is an ancillary system recording classical results. The state evolution is given by
\begin{equation}
	\mathcal{E}(\rho) = \sum_n E_n \rho E_n \otimes \ketbra{n}.
\end{equation}
Measurement result $n$ occurs with probability $\tr(\rho E_n)$.

The von Neumann measurement is a special case of PVMs. For a general PVM, several basis projectors may correspond to a same classical outcome, which represents a degeneracy phenomenon. In a similar fashion, we consider a generalized state preparation process. Given a set of projectors that span the entire Hilbert space, $\{E_n| n\in[d']\}$, we define the associated state preparation to be a map from $[d']$ to $\mathcal{D}(\mathcal{H})$, mapping $n \in [d']$ to $\frac{E_n}{\tr(E_n)}$. For $d'=d$, where the projectors form an orthonormal basis of the Hilbert space, this process is the usual state preparation.

Using the generalized measurement and state preparation, with respect to a set of projectors $\{E_n| n\in[d']\}$ on Hilbert space $\mathcal{H}$, we define quantum channel $\mathcal{T}_E$, which acts on a state $\rho\in\mathcal{D}(\mathcal{H})$ as
\begin{equation}
\mathcal{T}_E(\rho) = \sum_n \tr(\rho E_n) \frac{E_n}{\tr(E_n)}.
\end{equation}
The channel has the following operational interpretation: measure $\rho$ with $\{E_n\}$, obtain result $n$, and then prepare state $\frac{E_n}{\tr(E_n)}$ based on $n$. Note that when all of projectors $E_n$ are rank-one, that is, $\{E_n| n\in[d']\} = \{\ketbra{i}| i\in[d]\}$, $\mathcal{T}_E$ reduces to the dephasing operation on the basis $\{\ket{i}\}$. For a general case, $\mathcal{T}_E$ is different from the block-dephasing operation defined by the projectors~\cite{Aberg2006BlockCoherence}, \begin{equation}
  \Delta_E(\rho) = \sum_n E_n \rho E_n.
\end{equation}
Still, measuring the state $\mathcal{T}_E(\rho)$ with the PVM $\{E_n\}$ gives the same statistics as directly measuring $\rho$, that is, $\tr[\mathcal{T}_E(\rho) E_n] = \tr(\rho E_n)$.

Using the above ingredients, we define CRO in the extended scenarios of general PVMs. The formal definitions are the same as Definition~\ref{def:CRO} except for that the state preparation and measurement may not be rank-one. Here, we only discuss the extensions of cqCRO, qqCRO, and qcCRO, denoted as $\mathrm{cqCRO}_{\{E_n\}}$, $\mathrm{qqCRO}_{\{E_n\}}$, and $\mathrm{qcCRO}_{\{E_n\}}$ if the associated set of projectors is $\{E_n|n\in[d']\}$. By definition, we have $\mathrm{qqCRO}_{\{E_n\}}\subset \mathrm{cqCRO}_{\{E_n\}}$ and $\mathrm{qqCRO}_{\{E_n\}}\subset \mathrm{qcCRO}_{\{E_n\}}$. Similar to the dephasing operation $\Delta$, $\mathcal{T}_E$ plays an important role in expressing extensions of CROs. The result is formalized by the following theorem.

\begin{theorem}\label{thm:CROextequiv}
For any $O\in \mathrm{CPTP}$,
\begin{itemize}
\item $O\in\mathrm{cqCRO}_{\{E_n\}}\Leftrightarrow O\circ\mathcal{T}_E = \mathcal{T}_E\circ O\circ \mathcal{T}_E\Leftrightarrow \exists O^{\prime}\in\mathrm{CPTP}, O\circ \mathcal{T}_E = \mathcal{T}_E\circ O^{\prime}\circ \mathcal{T}_E$;
\item $O\in\mathrm{qqCRO}_{\{E_n\}}\Leftrightarrow O = \mathcal{T}_E\circ O\circ \mathcal{T}_E\Leftrightarrow \exists O^{\prime}\in\mathrm{CPTP}, O = \mathcal{T}_E\circ O^{\prime}\circ \mathcal{T}_E$;
\item $O\in\mathrm{qcCRO}_{\{E_n\}}\Leftrightarrow \mathcal{T}_E\circ O = \mathcal{T}_E\circ O\circ \mathcal{T}_E\Leftrightarrow \exists O^{\prime}\in\mathrm{CPTP}, \mathcal{T}_E\circ O = \mathcal{T}_E\circ O^{\prime}\circ \mathcal{T}_E$.
\end{itemize}
\end{theorem}

Theorem~\ref{thm:CROextequiv} is a generalization of Theorem~\ref{thm:CROequiv} and the proof ideas of the two are the same, with details shown in Appendix~\ref{thmproof:CROequiv}. Then, we define the sets of operations for extensions of CROs,
\begin{align}
	\label{eq:cqCROext} &\mathrm{cqCRO}_{\{E_n\}} = \{O\in \mathrm{CPTP}|O\circ\mathcal{T}_E = \mathcal{T}_E\circ O\circ \mathcal{T}_E\},\\
	\label{eq:qqCROext} &\mathrm{qqCRO}_{\{E_n\}} = \{O\in \mathrm{CPTP}|O = \mathcal{T}_E\circ O\circ \mathcal{T}_E\},\\
	\label{eq:qcCROext} &\mathrm{qcCRO}_{\{E_n\}} = \{O\in \mathrm{CPTP}|\mathcal{T}_E\circ O = \mathcal{T}_E\circ O\circ \mathcal{T}_E\}.
\end{align}

In the study of block coherence, with the block-dephasing operation $\Delta_E$, one can define the set of maximally block incoherent operations~\cite{Bischof2019POVMCoherence}, $\{O\in\mathrm{CPTP} | O\circ \Delta_E = \Delta_E\circ O\circ \Delta_E\}$, as a generalization of MIO. Note that in this case, $\mathrm{cqCRO}_{\{E_n\}}$ is in general not equivalent to maximally block incoherent operations. Similarly, $\mathrm{qqCRO}_{\{E_n\}}$ and $\mathrm{qcCRO}_{\{E_n\}}$ are in general not equal to $\{O\in \mathrm{CPTP}|O = \Delta_E\circ O\circ \Delta_E\}$ and $\{O\in \mathrm{CPTP}|\Delta_E \circ O = \Delta_E\circ O\circ \Delta_E\}$, respectively.

\subsection{Extension of CRO with Unitary Transformation}
In the discussion above, we fix the state preparation and measurement. In practice, the state preparation and measurement could be selected from multiple choices. When we use classical processing to replace the quantum operation, we can prepare the state or measure the state with one of these choices instead of a fixed one. 

For simplicity, we assume the state preparation and measurement is rank-one. As we can apply unitary to change bases, the freedom to choose bases is equivalent to the freedom to apply unitary after state preparation and before measurement. The unitary is chosen from an ensemble $\mathcal{U}$, which is often limited by quantum devices in practice. The largest unitary ensemble is $\mathcal{U} = \textsf{U}_d$, where $\textsf{U}_d$ is the unitary group with dimension $d$. In reality, single-qubit operations are normally easy to implement. In this case, we are often interested in the local unitary ensemble, $\mathcal{U} = \textsf{U}^{\otimes n}_2$, where $n$ is the number of qubits in the quantum system.

Given the ensemble $\mathcal{U}$, we denote the extensions of cqCRO, qqCRO and qcCRO as $\mathrm{cqCRO}_{\mathcal{U}}$, $\mathrm{qqCRO}_{\mathcal{U}}$, and $\mathrm{qcCRO}_{\mathcal{U}}$, respectively, as shown in Figure~\ref{fig:CROu}. The three kinds of $\mathrm{CRO}_{\mathcal{U}}$ are defined as follows.
\begin{enumerate}
	\item
	Consider operations right after the state preparation, a $\mathrm{cqCRO}_{\mathcal{U}}$ can be realized by first processing the input classically and then preparing the state under a chosen basis.
	
	\item
	Consider operations with quantum input and quantum output. A $\mathrm{qqCRO}_{\mathcal{U}}$ can be realized by first measuring the input with a basis, processing the outcomes classically, and then preparing the output state under another possibly different basis.
	
	\item
	Consider operations right before measurement. A $\mathrm{qcCRO}_{\mathcal{U}}$ can be realized by measuring the quantum input and processing the outcomes classically.
\end{enumerate}
Note that the classical value can always be discriminated and copied, which means we can choose the state preparation basis based on the classical input and measurement outcomes. This is the origin of control unitary in Figure~\ref{fig:cqCROu} and~\ref{fig:qqCROu}. In our model, we suppose the quantum input is unknown and cannot be discriminated before measuring it, so the measurement basis cannot be chosen depending on the quantum input. That means in the first step of replacing $\mathrm{qqCRO}_{\mathcal{U}}$ or $\mathrm{qcCRO}_{\mathcal{U}}$, we need to select a fixed measurement independent of the input.

\begin{figure}[!htbp]
	\centering
	\subfigure[$\mathrm{cqCRO}_{\mathcal{U}}$]{
		\begin{minipage}[t]{0.4\linewidth}
			\centering
			\includegraphics[scale = 0.24]{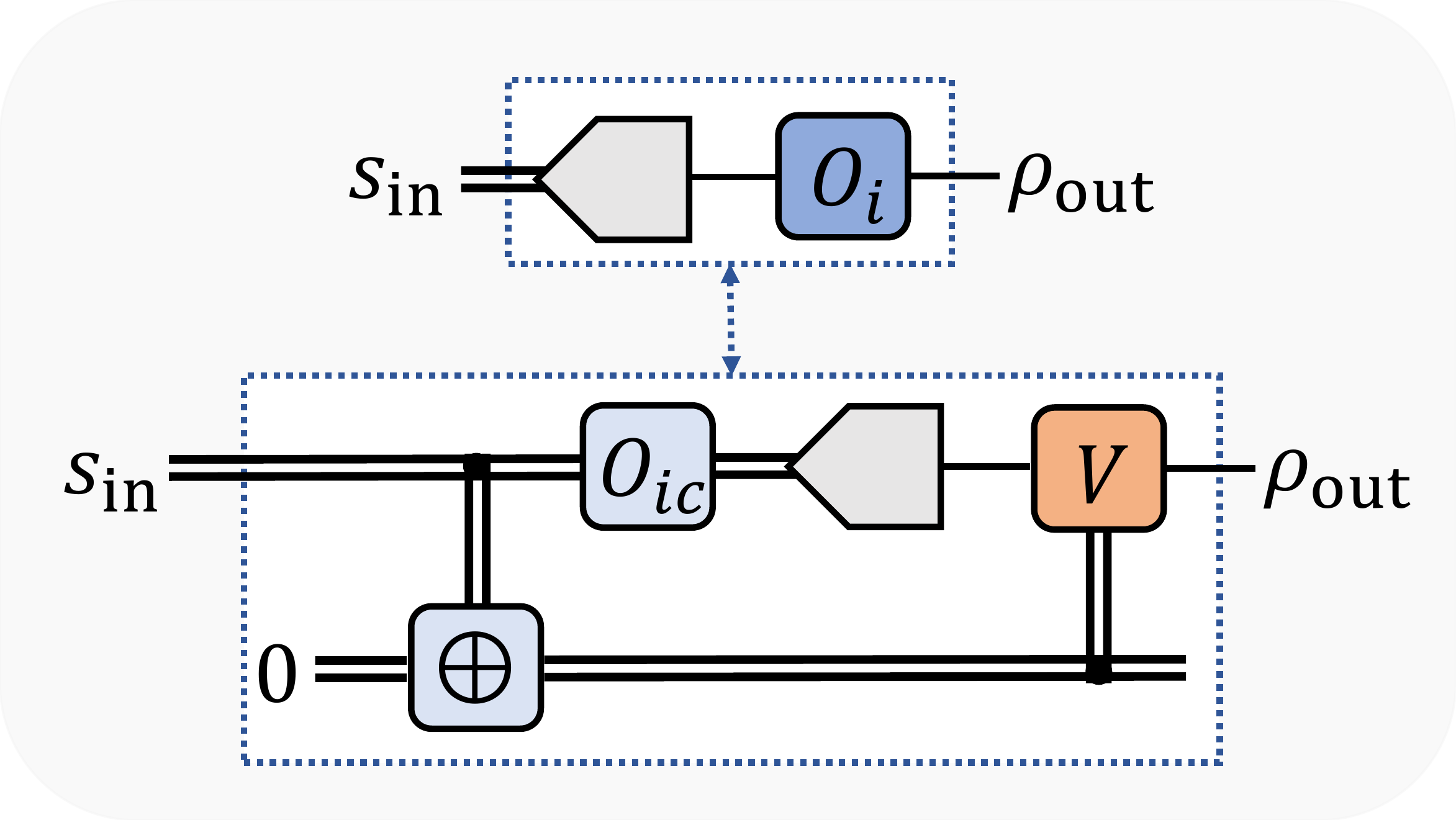}
			\label{fig:cqCROu}
%			\caption{cqCROu}
		\end{minipage}
	}
	\subfigure[$\mathrm{qqCRO}_{\mathcal{U}}$]{
		\begin{minipage}[t]{0.4\linewidth}
			\centering
			\includegraphics[scale = 0.24]{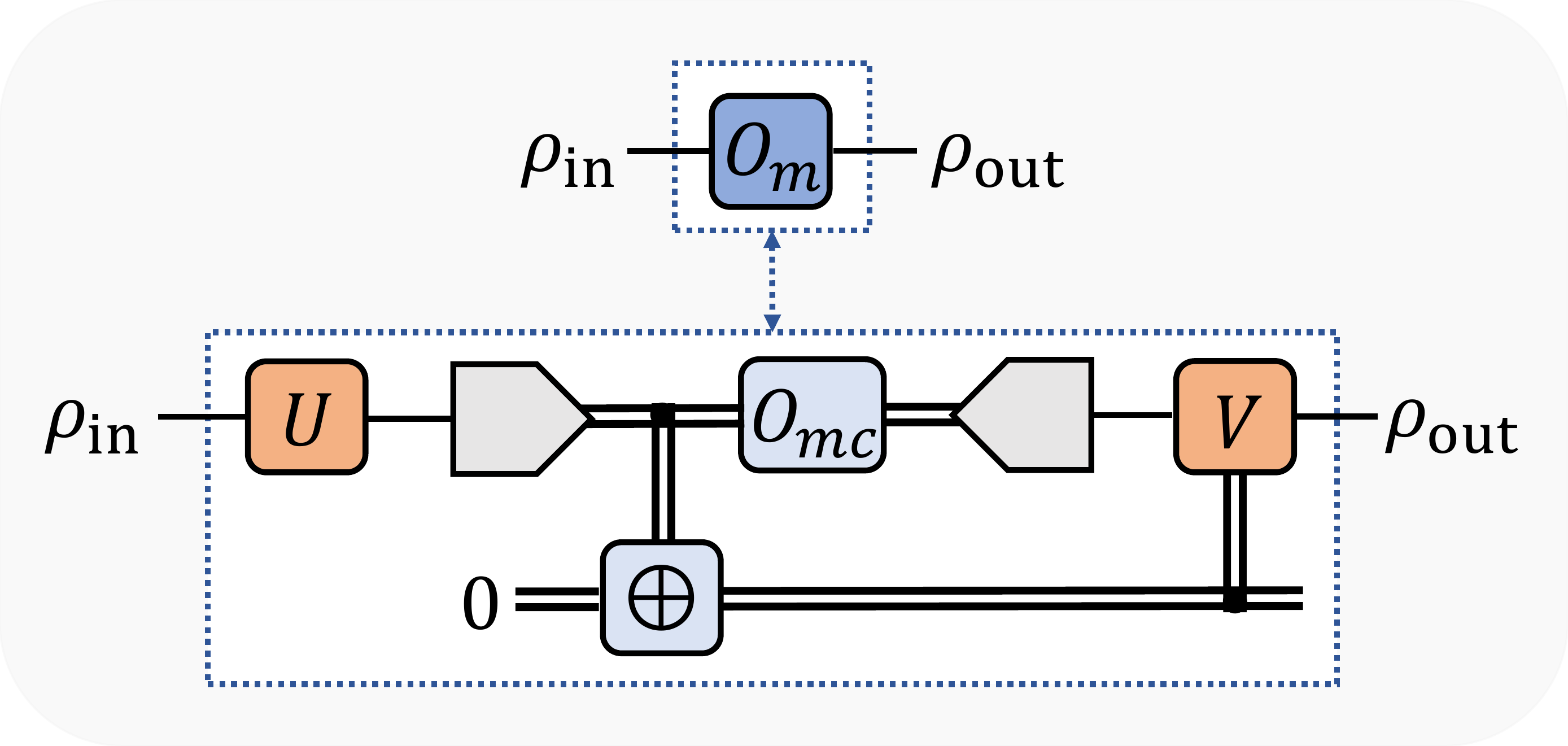}
			\label{fig:qqCROu}
%			\caption{qqCROu}
		\end{minipage}
	}
	\subfigure[$\mathrm{qcCRO}_{\mathcal{U}}$]{
		\begin{minipage}[t]{0.8\linewidth}
\centering
\includegraphics[scale = 0.26]{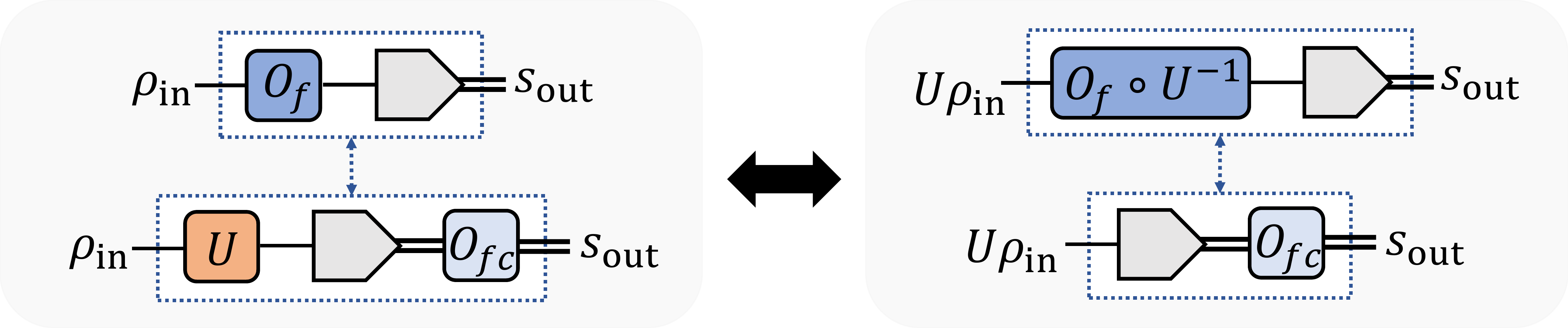}
\label{fig:qcCROu}
%			\caption{qcCROu}
\end{minipage}
}
\centering
\caption{The equivalent circuits for (a) $\mathrm{cqCRO}_{\mathcal{U}}$, (b) $\mathrm{qqCRO}_{\mathcal{U}}$, (c) $\mathrm{qcCRO}_{\mathcal{U}}$ when we relax the restriction of the state preparation and measurement bases. The notation $\oplus$ in (a) and (b) represents the module-$d$ summation, essentially a classical copy operation. (a) A $\mathrm{cqCRO}_{\mathcal{U}}$, $O_i$, can be realized by processing the input classically with $O_{ic}$, followed by the state preparation whose basis or unitary $V$ can depend on the classical input. (b) A $\mathrm{qqCRO}_{\mathcal{U}}$, $O_m$, can be realized by measuring the input, processing the measurement result classically with $O_{mc}$, and state preparation. The basis of state preparation or unitary $V$, can depend on the measurement result, while the measurement basis or unitary $U$ is independent of the quantum input. (c) A $\mathrm{qcCRO}_{\mathcal{U}}$ can be realized by measuring the input and processing the measurement result classically, where the measurement basis or unitary $U$ is independent of the quantum input. It is equivalent to the requirement of $\exists U\in\mathcal{U}$, $O_f\circ U^{-1}\in\mathrm{qcCRO}$.}\label{fig:CROu}
\end{figure}

From Figure~\ref{fig:qcCROu} we can see that an operation, $O$, is $\mathrm{qcCRO}_{\mathcal{U}}$ if and only if there exists unitary $U\in\mathcal{U}$ s.t., $O\circ U^{-1}\in \mathrm{qcCRO}$. The set of $\mathrm{qcCRO}_{\mathcal{U}}$ is the union set of $\{\mathrm{qcCRO}_{U}\}$ where $U\in\mathcal{U}$,
\begin{equation}\label{eq:qcCROu} \mathrm{qcCRO}_{\mathcal{U}} = \{O\in \mathrm{CPTP}|\exists U\in \mathcal{U}, \Delta\circ (O\circ U^{-1}) = \Delta\circ (O\circ U^{-1})\circ \Delta\}.
\end{equation}
We remark that unlike $\mathrm{qcCRO}_{\mathcal{U}}$, $\mathrm{cqCRO}_{\mathcal{U}}$ or $\mathrm{qqCRO}_{\mathcal{U}}$ is not a simple union of $\{\mathrm{cqCRO}_{U}\}$ or $\{\mathrm{qqCRO}_{U}\}$, $U\in\mathcal{U}$, due to the possible dependence of state preparation basis on the classical input or the measurement result.

From the operational meaning of $\mathrm{qqCRO}_{\mathcal{U}}$, we can find another representation. Any $\mathrm{qqCRO}_{\mathcal{U}}$ can be realized by measuring the input state $\rho$ with a basis $\{e_i\}$, transforming measurement result $i$ into $j$ with probability $T_{ji}$, and then preparing the state under another possibly different basis $\{f^i_j\}$,
\begin{equation}\label{eq:qqCROu}
	\begin{split}
		O_m(\rho) &= \sum_{i,j}T_{ji}\bra{e_i}\rho\ket{e_i} \ketbra{f^i_j}\\
		&= \sum_i \bra{e_i}\rho\ket{e_i} (\sum_j T_{ji}\ketbra{f^i_j}),
	\end{split}
\end{equation}
where $T_{ji}$ is the element of a stochastic matrix, $\{\ket{f^i_j}\}$ reflects the dependence between the state preparation basis and the measurement result $i$. From Eq.~\eqref{eq:qqCROu} we can see that any $\mathrm{qqCRO}_{\mathcal{U}}$ is an EB channel~\cite{Horodecki2003EBChannel}, which measures the state with a positive operator-valued measure (POVM) followed by state preparation. 

A classical input can always be modelled as a quantum input followed with the computational basis measurement as shown in Figure~\ref{fig:cq2qqCROu}. We can see that if an operation, $O_i$, is a $\mathrm{cqCRO}_{\mathcal{U}}$, then $O_i\circ \Delta$ is a $\mathrm{qqCRO}_{\mathcal{U}}$. From a similar argument of $\mathrm{qqCRO}_{\mathcal{U}}$, we obtain that any $\mathrm{cqCRO}_{\mathcal{U}}$, $O_i$, satisfies,
\begin{equation}\label{eq:cqCROu}
	\begin{split}
		O_i\circ \Delta(\rho) &= \sum_{k,j}T_{jk}\bra{k}\rho\ket{k} \ketbra{f^k_j}\\
		&= \sum_k \bra{k}\rho\ket{k} (\sum_j T_{jk}\ketbra{f^k_j}),
	\end{split}
\end{equation}
where $T_{jk}$ is the element of a stochastic matrix, $\{\ket{k}\}$ is the computational basis, $\{\ket{f^k_j}\}$ is a basis for state preparation.

\begin{figure}[!htbp]
	\centering
\resizebox{7cm}{!}{\includegraphics{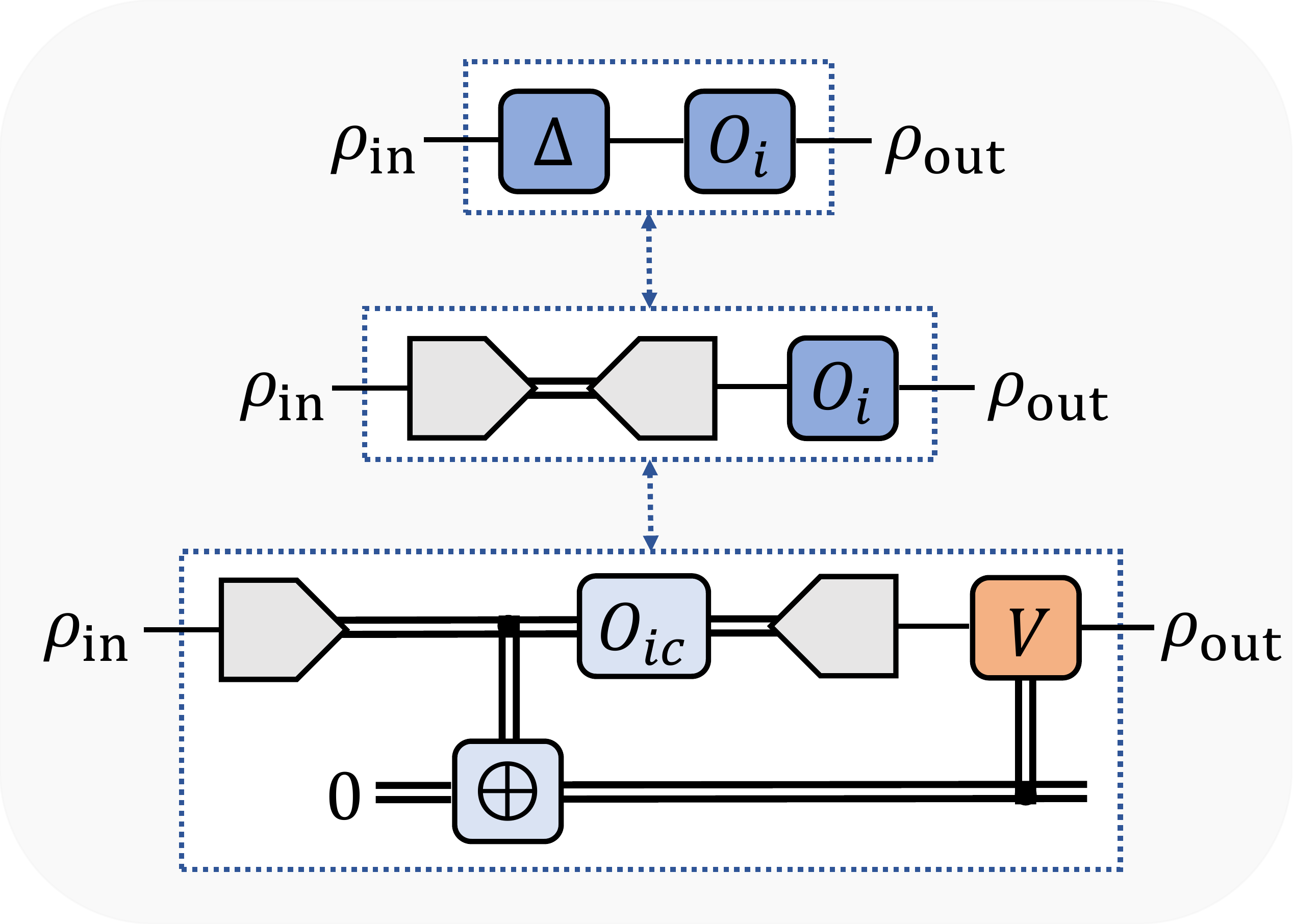}}
	\caption{A classical input $s_{\mathrm{in}}$ is equivalent to a quantum input $\rho_{\mathrm{in}}$ followed with the computational basis measurement. Here, the basis of state preparation right before $O_i$ in the second frame is set as the computational basis. Then, the equivalence in Figure~\ref{fig:cqCROu} turns into the equivalence of the second and third lines. If $O_i$ is a $\mathrm{cqCRO}_{\mathcal{U}}$, then $O_i\circ \Delta$ is a $\mathrm{qqCRO}_{\mathcal{U}}$. Note that the measurement followed by state preparation is a dephasing operation.}
	\label{fig:cq2qqCROu}
\end{figure}

Interestingly, if $\mathcal{U} = \textsf{U}_d$, then any CPTP map is a $\mathrm{cqCRO}_{\mathcal{U}}$. The reason is that we can evaluate the output quantum state $\rho_{\mathrm{out}}$ after reading the classical input $s_{\mathrm{in}}$. With the freedom to choose arbitrary unitary operations, one can prepare $\rho_{\mathrm{out}}$ directly.

\subsection{Virtual Clifford Gate from Classical Replacement}
Now, we discuss the application of classical replacement. Consider a task in VQA that one needs to estimate the lowest energy level of a given Hamiltonian, $H$~\cite{Cerezo2021Variational}. The basic idea of a VQA is preparing a parametrized state $\rho$ and estimating its energy $\tr(\rho H)$. Then, one can minimize this value to find the approximate lowest energy by adjusting the parameters of $\rho$. In general, an $n$-qubit quantum state, $\rho$, is prepared with a number of parametrized gates,
\begin{equation}
\rho = C(\alpha)U(\beta)\cdots V(\gamma) \rho_0 V^{\dagger}(\gamma)\cdots U^{\dagger}(\beta) C^{\dagger}(\alpha),
\end{equation}
where $\rho_0$ is the initial state and $\alpha, \beta, \cdots, \gamma$ are parameters for the VQA. Then, one performs the measurement on $\rho$ to evaluate $\tr(\rho H)$. In practice, Pauli measurements are often favoured and one often decomposes $H$ into the sum of $n$-qubit Pauli operators, $H = \sum_{i\in I} c_iP_i, P_i\in \textsf{P}_n$ where $\textsf{P}_n$ is the $n$-qubit Pauli group and $I \subseteq [4^n]$ is an index set for $\textsf{P}_n$. Then, estimating $\tr(\rho H)$ becomes evaluating values of $\tr(\rho P_i)$.

In general, one needs to perform a PVM associated with $P_i$ on $\rho$ to obtain $\tr(\rho P_i)$, which we call $P_i$-measurement. For some gates $C(\alpha)$ in the state preparation of $\rho$, $C(\alpha)$ followed with $P_i$-measurement can be replaced by $P_j$-measurement followed with classical processing, where $P_i,P_j\in \textsf{P}_n$. Then, we can replace $C(\alpha)$ with classical processing and skip it in the quantum circuit preparing $\rho$. Following the same arguments in previous subsections, the set of CRO in this case is given by
\begin{equation}
\mathcal{R} = \{ O\in \mathrm{CPTP} | \exists j\in [4^n], \forall i\in I, \mathcal{T}_i \circ O = \mathcal{T}_i \circ O \circ \mathcal{T}_j \}.
\end{equation}
Here, $\mathcal{T}_i(\rho) = \frac{1}{2^{n-1}}(\tr(P_i^+ \rho)P_i^+ + \tr(P_i^- \rho)P_i^-)$, where $P_i^+ = (\mathbb{I} + P_i) / 2$ and $P_i^- = (\mathbb{I} - P_i) / 2$ are the projectors to the eigenspaces of $P_i$ with eigenvalues $+1$ and $-1$, respectively. Note that $\forall i\in I$, the gates in $\mathcal{R}$ before $P_i$-measurement can be replaced classically. We provide the detailed derivation and discussion in Appendix~\ref{appendsc:vqa}. Interestingly, we can verify that the $n$-qubit Clifford group $\textsf{C}_n$ is always a subset of $\mathcal{R}$ regardless of the choice of the index set $I$. In fact, we can first find a subset of $\mathcal{R}$,
\begin{equation}
\begin{split}
\mathcal{R}' &= \{ O\in \mathrm{CPTP} | \exists C\in \textsf{C}_n, \forall i\in I, \mathcal{T}_i \circ (O \circ C^{-1}) = \mathcal{T}_i \circ (O \circ C^{-1}) \circ \mathcal{T}_i \}\\
&= \{ O\in \mathrm{CPTP} | \exists C\in \textsf{C}_n, \forall i\in I, \mathcal{T}_i \circ O  = \mathcal{T}_i \circ O \circ (C^{-1} \circ \mathcal{T}_i \circ C) \}\\
&\subseteq \mathcal{R}.
\end{split}
\end{equation}
The third line comes from the fact that any Clifford gate $C$ satisfies $C^{-1} \textsf{P}_n C = \textsf{P}_n$, which means $\exists j\in [4^n]$, $\mathcal{T}_j = C^{-1} \circ \mathcal{T}_i \circ C$. For any Clifford gate $C\in \textsf{C}_n$, $\exists C\in\textsf{C}_n$, $\mathcal{T}_i \circ (C \circ C^{-1}) = \mathcal{T}_i \circ (C \circ C^{-1}) \circ \mathcal{T}_i$. Thus, $\textsf{C}_n \subseteq \mathcal{R}' \subseteq \mathcal{R}$. Any Clifford gate before Pauli measurements can be replaced and requires no real implementation, which has been discovered and utilised from the perspective of the quantum evolution in the Heisenberg picture~\cite{ZhongXia2021SHVQA}.

Essentially, the classical replaceability of Clifford gates comes from the commutation relation in the Clifford algebra. This fact has been observed and widely used in quantum error correction. For example, in the Clifford+$T$ model~\cite{Litinski2019gameofsurfacecodes}, one would remove all Clifford gates ahead of the final measurements. If one can perform mutually commuted Pauli measurements, then the Clifford gates before measurements can be absorbed. Thus, one only needs to implement multi-qubit $\frac{\pi}{8}$-rotations in this model. Another example utilising the Clifford algebra to reduce quantum gates is ``Pauli frame''~\cite{Knill2005Pauliframe,Chamberland2018faulttolerant,Suzuki2022ErrorMitigation}. Through updating the so-called ``Pauli frame'' with classical processing, there is no need for explicitly implementing Pauli gates before a Clifford gate.

The CRO set, $\mathcal{R}$, as well as $\mathcal{R}'$ are in general larger than $\textsf{C}_n$ since the classical replacement does not necessarily rely on the Clifford algebra. Take $H = Z^{\otimes 3}$ as an example where $Z$ is the Pauli-$Z$ gate. Then, the non-Clifford gate controlled-controlled-NOT, or CCX, is also an element in $\mathcal{R}'$ since $\mathcal{T}_i \circ \mathrm{CCX} = \mathcal{T}_i \circ \mathrm{CCX} \circ \mathcal{T}_i$. Note that this can be extended to the case of the Toffoli gate with more than two control qubits. It means that we can virtually apply a number of quantum gates beyond the Clifford group. This might be an approach to improving the expressibility of the variational quantum circuit while saving the consumption of the experimental resource.

\section{Characterization of Irreplaceability}\label{sc:Resource}
The replaceability with classical processing of CRO motivates us to view the irreplaceability as a kind of quantum resource~\cite{Streltsov2017CoherenceRMP,chitambar2019quantum,liu2020operational,gour2021entropy}. We can establish a channel resource theory to quantitatively study the potential quantum advantage brought by irreplaceability. We take CRO as the set of free channels and any quantum operation outside this set contains the resource of irreplaceability. Depending on the nature of the channel input and output, we can choose different types of CROs and specify a corresponding resource theory. In this work, we take qcCRO as an example.

To establish a channel resource theory for irreplaceability, we first need to specify the free channels and free superchannels~\cite{liu2020operational}. The set of free channels is naturally given by Eq.~\eqref{eq:qcCRO}. We consider the set of resource non-generating (RNG) superchannels to be the set of free superchannels $\mathcal{F}$,
\begin{equation}
	\mathcal{F} = \mathrm{RNG} = \{\Lambda| \forall \mathcal{M}\in \mathrm{qcCRO}, \Lambda(\mathcal{M})\in \mathrm{qcCRO}\}.
\end{equation}
To quantify the amount of irreplaceability of a channel, we can utilize the Choi-state representation of the channel. Given a channel $\mathcal{N}$ acting on $\mathcal{D}(\mathcal{H})$, the corresponding Choi-state is
\begin{equation}
\Phi_{\mathcal{N}} = I\otimes \mathcal{N} (\ketbra{\Phi^+}),
\end{equation}
where $\ket{\Phi^+} = \frac{1}{\sqrt{d}}\sum_{i=0}^{d-1} \ket{ii}$ is the maximally entangled state on $\mathcal{H}\otimes \mathcal{H}$. We define two measures to characterize irreplaceability. One is relative entropy~\cite{Vedral1997Entanglement} of irreplaceability with details shown in Appendix~\ref{append:rel}. The other is robustness~\cite{Vidal1999Robustness} of irreplaceability defined as follows.

\begin{definition}[Robustness of Irreplaceability]
Given a channel $\mathcal{N}\in CPTP$, the robustness of irreplaceability of $\mathcal{N}$ is
\begin{equation}\label{eq:RobofIrDef1}
R(\mathcal{N})
= \min_{\mathcal{M}\in \mathrm{CPTP} } \left\{s\geq 0 \bigg| \frac{\Phi_{\mathcal{N}}+s\Phi_{\mathcal{M}}}{1+s}\in \mathbf{F} \right\},
\end{equation}
where $\mathbf{F} = \{\Phi_{\mathcal{M}}\big| \mathcal{M}\in \mathrm{qcCRO}\}$.
\end{definition}

The robustness of irreplaceability in Eq.~\eqref{eq:RobofIrDef1} has several equivalent definitions, as shown by the following lemma.

\begin{lemma}\label{lemma:equiv}
The robustness of irreplaceability of a channel $\mathcal{N}$ can be equivalently given by
\begin{align}
\label{eq:RobofIrDef2} R(\mathcal{N}) &= \min_{\mathcal{M}\in \mathrm{CPTP}} \left\{s\geq 0 \bigg| \frac{\Phi_{\Delta\circ\mathcal{N}}+s\Phi_{\Delta\circ\mathcal{M}}}{1+s}\in \mathbf{F}'\right\}\\
\label{eq:RobofIrDef3} &= \min_{\mathcal{M}\in \mathrm{CPTP} } \left\{s\geq 0 \bigg| \frac{\mathcal{N}+s\mathcal{M}}{1+s}\in \mathrm{qcCRO}\right\}\\
\label{eq:RobofIrDef4} &= \min_{\mathcal{M}\in \mathrm{CPTP} } \left\{s\geq 0 \bigg| \frac{\Delta\circ\mathcal{N}+s\Delta\circ\mathcal{M}}{1+s}\in \mathrm{qcCRO}\right\}\\
\label{eq:RobofIrDef5} &= \min_{\mathcal{M}\in \mathrm{CPTP} } \left\{s\geq 0 \bigg| \frac{\Delta\circ\mathcal{N}+s\mathcal{M}}{1+s}\in \mathrm{qcCRO}\right\},
\end{align}
where $\mathbf{F}'=\{\Phi_{\Delta\circ\mathcal{M}}\big| \mathcal{M}\in \mathrm{qcCRO}\}$.
\end{lemma}

We leave the proof of the lemma in Appendix~\ref{ProofLemmaEquiv}. From Lemma~\ref{lemma:equiv}, we can see that $R(\mathcal{N})=R(\Delta\circ\mathcal{N})$. The robustness is a valid measure, which vanishes to zero only for qcCRO and enjoys the properties of monotonicity under free superchannels and convexity under channel mixing, as shown by the following lemma.

\begin{lemma}\label{lemma:robust}
The robustness of irreplaceability in Eq.~\eqref{eq:RobofIrDef1}, $R(\mathcal{N})$, has the following properties:
\begin{enumerate}
\item
\emph{Monotonicity: }
$\forall\mathcal{N}\in \mathrm{CPTP}$, $\forall\Lambda\in \mathcal{F}$,
\begin{equation}
R(\Lambda(\mathcal{N}))\leq R(\mathcal{N}).
\end{equation}

\item
\emph{Convexity: }
Given an index set $\mathcal{I}$, $\forall\mathcal{N}_{i\in\mathcal{I}}\in \mathrm{CPTP}$, $\forall\{p_i\}_{i\in\mathcal{I}}$ such that $p_i\geq0,\sum_{i\in\mathcal{I}}p_i=1$,
\begin{equation}
R\left(\sum_{i}p_i\mathcal{N}_i\right)\leq \sum_{i}p_i R(\mathcal{N}_i).
\end{equation}
\end{enumerate}
\end{lemma}

We leave the proof of the lemma in Appendix~\ref{ProofLemmaRobust}. Intuitively, these properties originate from the convexity of the set of free channels. The robustness can hence be viewed as a geometric measure for irreplaceability that quantifies the distance between the considered channel and the set of qcCRO.

When studying a large quantum system composed of several parts, if the action of an operation is restricted to one part of the system, we expect its robustness of irreplaceability to remain the same when considering its trivial extension to the whole system. This is indeed the case as shown in the following lemma. We leave the proof in Appendix~\ref{ProofLemmaExtension}.

\begin{lemma}\label{lemma:extension}
The robustness of irreplaceability, $R(\mathcal{N})$, satisfies $R(\mathcal{N}) = R(\mathcal{N}\otimes I)$, where $I$ is the identity operation of an ancillary system with an arbitrary finite dimension.
\end{lemma}

Similar to other robustness-type measures~\cite{uola2019quantifying}, the calculation of the robustness of irreplaceability can be cast into a conic programming problem and solved efficiently. We provide the detailed method for its calculation in Appendix~\ref{Append:RobustCal}. Here, we present some numerical results. Consider the family of single-qubit gates, $U(\theta) = \cos \theta Z + \sin \theta X,\theta\in[0,\pi/2]$, and assign the computational basis to be the eigenvectors of Pauli operator $Z$. In the case of qcCRO, when $\theta=0$, the gate becomes Pauli operator $Z$, which can be omitted when followed with the computational-basis measurement. When $\theta=\pi/2$, the gate becomes Pauli operator $X$, which can be replaced by a classical flip operation after the measurement. In cases other than these two extremes, the gate has non-zero irreplaceability. In Figure~\ref{fig:ResourceMeasure}, we plot the amount of irreplaceability with respect to the gate set parameter, $\theta$. In addition to the robustness measure, we also depict the relative entropy of the irreplaceability of the gates.

\begin{figure}[!htbp]
\centering
\resizebox{9cm}{!}{\includegraphics{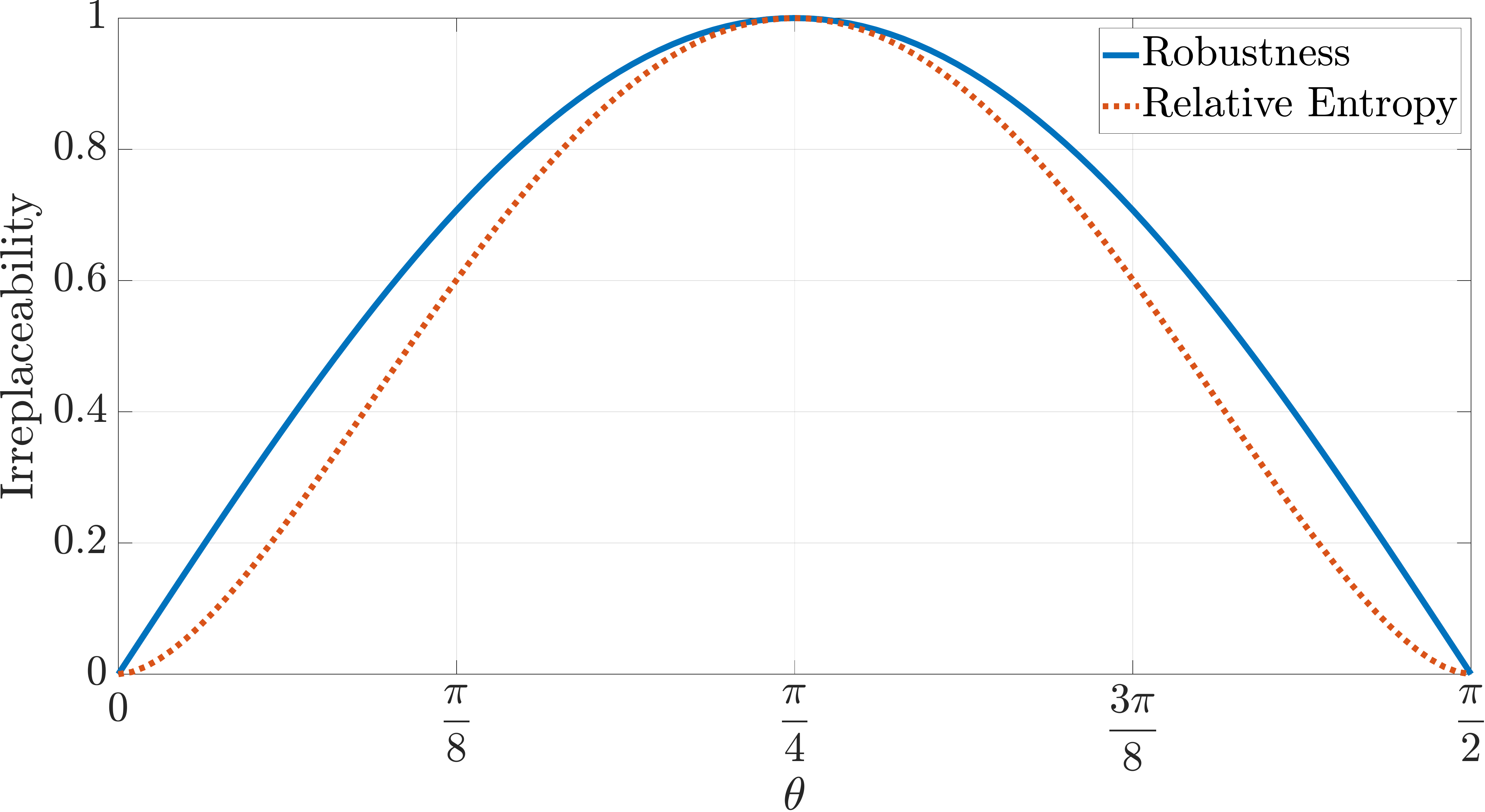}}
	\caption{Irreplaceability of the single-qubit gates family, $\{U(\theta) = \cos \theta Z + \sin \theta X|\theta\in[0,\pi/2]\}$, as manifested in robustness and relative entropy measures. The two ends, when $\theta$ equals $0$ and $\pi/2$, represent two qcCROs, $Z$, and $X$, respectively. While $0<\theta<\pi/2$, $U(\theta)$ is not classically replaceable. The amount of irreplaceability achieves the highest in the case that $\theta = \pi/4$, or equivalently, $U(\theta)$ is a Hadamard gate.}
	\label{fig:ResourceMeasure}
\end{figure}

Interestingly, the robustness of irreplaceability has an operational meaning: it measures the advantage that a channel can provide in a non-local game~\cite{Takagi2019Robustness,Yuan2021Memory} of state discrimination. To be specific, we define a bipartite non-local game in Box~\ref{box:nonlocal}.

\begin{game}{box:nonlocal}{: Non-local Game}
A non-local game $\mathcal{G} =\{ \{\alpha_{ij}\}, \{\sigma_i\}, \{\ketbra{j}\} \}$ for two parties, Alice and Bob, is composed of the following items:
\begin{itemize}
  \item $\{\alpha_{ij}\}$ --- payoff values, $\alpha_{ij}\in\mathbb{R},\forall i,j$
  \item $\{\sigma_i\}$ --- a set of quantum states Alice prepares
  \item $\{\ketbra{j}\}$ --- the computational basis measurement Bob performs
\end{itemize}

\tcblower
\begin{enumerate}
\item Alice randomly and uniformly selects a state from $\{\sigma_i\}$ and sends it to Bob via the quantum channel $\mathcal{N}$.

\item Bob performs the computational basis measurement $\{\ketbra{j}\}$ on his received state. If he obtains the result $j$, then he guesses the received state to be $\sigma_j$. If Alice sends $\sigma_i$, Alice and Bob obtain the corresponding payoff $\alpha_{ij}$.

\item Alice and Bob repeat the game for sufficiently many times and calculate the average payoff value.

\end{enumerate}
\end{game}

In the case where Alice sends the state $\sigma_i$, the probability that Bob takes a guess of $\sigma_j$ is $\tr(\mathcal{N}(\sigma_i) \ketbra{j})$. On average, the performance that Alice and Bob can obtain in this game is evaluated by the expected payoff function,
\begin{equation}
  p\left(\mathcal{N}, \mathcal{G}\right) = \sum_{i,j} \alpha_{ij} \tr(\mathcal{N}(\sigma_i) \ketbra{j}).
\end{equation}
For simplicity, we call it the performance of $\mathcal{N}$ in the non-local game $\mathcal{G}$.
Without loss of generality, we can require that $\forall\mathcal{M}\in \mathrm{qcCRO}$, the payoff values $\alpha_{ij}$ satisfy that the performance of all qcCROs is non-negative and bounded by $1$. Under this constraint, we have the following theorem showing the advantage brought by an irreplaceable quantum channel.

\begin{theorem}\label{thm:nonlocal}
The best advantage a quantum channel can provide over all qcCROs and all non-local games is given by $1 + R(\mathcal{N})$:
\begin{equation}\label{eq:advantage2robust}
\begin{split}
&\max_{\{\alpha_{ij}\}, \{\sigma_i\}} \min_{\mathcal{M}\in\mathrm{qcCRO}} \frac{p\left(\mathcal{N}, \mathcal{G}\right)}{\ p\left(\mathcal{M}, \mathcal{G}\right)} = 1 + R(\mathcal{N}),\\
&\text{s.t. } 0\leq p(\mathcal{M}, \mathcal{G})\leq 1,\forall \mathcal{M}\in \mathrm{qcCRO}.
\end{split}
\end{equation}
\end{theorem}

The proof of Theorem~\ref{Append:ProofThmNonlocal} follows the idea of duality in conic programming~\cite{Takagi2019Robustness,uola2019quantifying}. We leave the detailed proof in Appendix~\ref{thmproof:resource}. Here, we discuss the operational meaning of the result. A qcCRO followed with a computational basis measurement is equivalent to applying a computational basis measurement followed with classical processing. As a result, Bob cannot distinguish the off-diagonal terms of $\{\sigma_i\}$ with a qcCRO. If the considered quantum channel is irreplaceable, when optimizing over all possible non-local games to maximize its advantage, we may choose a set of states $\{\sigma_i\}$ that vary mainly in the off-diagonal terms. The replaceable channels would fail in distinguishing the difference in these elements. On the other hand, a higher value of the robustness of irreplaceability brings a stronger ability of a channel to probe the off-diagonal terms of $\{\sigma_i\}$ and hence a higher performance in the non-local game.

In the non-local game as shown in Box~\ref{box:nonlocal}, Bob can only perform computational basis measurement to distinguish quantum states. We can generalize the non-local game by considering other measurement settings. In fact, for the case where Bob can perform arbitrary POVM to distinguish quantum states, the best advantage a quantum channel can provide over all qcCROs and all non-local games is also $1 + R(\mathcal{N})$. This result can be viewed as a special case of Ref.~\cite{Takagi2019Robustness}. Another interesting case is where Bob can choose arbitrary local basis measurements. This is often the scenario in practical implementation of quantum information processing protocols, like shadow tomography~\cite{huang2020shadow}. In this case, we could choose the convex hull of $\mathrm{qcCRO}_{\mathcal{U}}$ as the set of free channels, where $\mathcal{U}$ equals $\textsf{U}_2^{\otimes n}$ and represents the freedom to choose measurement bases. We expect to obtain similar results like Theorem~\ref{thm:nonlocal} when considering the advantage brought by an irreplaceable quantum channel and leave it for future works.

\section{Conclusion}\label{sc:Conclusion}
In this work, we define the concept of CROs in four scenarios. Depending on the feature of the input and output of an operation, we classify CROs into four types. Among these sets, ccCRO is the largest, since it composes of all quantum operations, while qqCRO is the smallest, being a subset of each of the other three sets. We provide necessary and sufficient criteria to determine whether an operation is a CRO and present its corresponding classical processing. Interestingly, for the special cases of unitary operations, namely, cqCRU, qqCRU, and qcCRU, one only needs deterministic classical processing for replacement. Furthermore, we discuss two extensions of CRO, where the state preparation and measurement may not be rank-one and fixed. As an application, we show that a number of quantum gates beyond the Clifford group can be replaced in VQA, unveiling an approach to enhancing the expressibility of a variational quantum circuit with classical replacement.

From a theoretical view, the clarification of replaceability and irreplaceability manifests a difference between classical and quantum operations. Focusing on the set of qcCRO, we establish a resource theory framework for the study of irreplaceability with classical operations. We propose two measures, namely, robustness and relative entropy, to quantify the resource. An interesting discovery is that the robustness measure quantifies the quantum advantage of an operation over all the qcCROs in a non-local game. Besides, we find that cqCRO and qcCRO are equivalent to MIO and CNAO in the resource theory of coherence, respectively, while qqCRO is a subset of both DIO and EB channels. The relation reveals the connection between irreplaceability, coherence, and entanglement.

Note that, unlike classical simulation, we are only concerned with whether a quantum operation can be replaced regardless of the consumption of classical computing resources. This is the case in many subjects like the security analysis of QKD. It may be practically interesting to study the circuit complexity issues~\cite{Eisert2021Complexity} in the classical replacement. Also, verifying whether a quantum gate, especially a large circuit composed of many quantum gates, is classically replaceable is important in quantum computing.

There are some other interesting future directions. One direction is extending the concept of CRO from finite dimensions to infinite dimensions. We expect such studies to benefit continuous-variable quantum information processing and inspire new perspectives on the non-classicality in the continuous-variable regime. Another topic is studying the problem of classical replacement in a non-Markovian evolution, that is, the evolution between two different times may not be a CPTP map. It is also interesting to further explore the resource theory of irreplaceability. As in other resource theories, we expect to witness more valid measures with operational meanings for irreplaceability.

\begin{acknowledgements}
We thank Pei Zeng and Junjie Chen for the helpful discussions. This work was supported by the National Natural Science Foundation of China Grants No.~11875173 and No.~12174216 and by the National Key Research and Development Program of China Grants No.~2019QY0702 and No.~2017YFA0303903.
\end{acknowledgements}

%%%%%%%%%%%%%%%%%%%%%%%%%%%%%%%%%%%%%%%%
% choose a style
%\bibliographystyle{ieeetr}
%\bibliographystyle{unsrt}
\bibliographystyle{apsrev}
%\bibliographystyle{unsrtnat}
%%%%%%%%%%%%%%%%%%%%%%%%%%%%%%%%%%%%%%%%

%%%%%%%%%%%%%%%%%%%%%%%%%%%%%%%%%%%%%%%%
% choose a .bib file
\bibliography{bibCO.bib}
%%%%%%%%%%%%%%%%%%%%%%%%%%%%%%%%%%%%%%%%

\newpage
\appendix
\section{Resource Theory}\label{appendsc:resource}
Quantum resource theory (QRT) studies how to characterize the resource stored in a quantum state, like entanglement~\cite{Vedral1997Entanglement} and coherence~\cite{Baumgratz2014Coherence}. It also provides a tool to explore the problem of the interconversion between different resources under specific restrictions. In this section we will review the framework of resource theory~\cite{Brand2015Reversible} and introduce the resource theory of coherence~\cite{Baumgratz2014Coherence} as well as the channel resource theory.

\subsection{Framework of Resource Theory}\label{appendssc:resourceframework}
For any QRT, we first point out the following as three main ingredients: \textit{resource, free states, and free operations}. The resource like coherence is the quantity we characterize in the resource theory. The term ``free'' means that this kind of states or operations can be obtained at no cost. The state not free is a resourceful state. The three ingredients are not independent to each other. They satisfy \textit{the free operations postulate (FOP)}~\cite{Brand2015Reversible}, that is, any free operation cannot transform a free state to a resourceful state. Denote the Hilbert space as $\mathcal{H}$, the states on $\mathcal{H}$ as $\mathcal{D}(\mathcal{H})$, the set of free states on $\mathcal{D}(\mathcal{H})$ as $\mathcal{F}$. Then for any free operation $\Lambda$,
\begin{equation}
	\Lambda(\sigma)\in \mathcal{F}, \forall \sigma\in \mathcal{F}.
\end{equation}
From this postulate, we could define a set of operations named resource non-generating operations (RNG), containing all the operations satisfying FOP:
\begin{equation}
\mathrm{RNG} = \{\Lambda\in \mathrm{CPTP}|\Lambda(\sigma)\in \mathcal{F}, \forall \sigma\in \mathcal{F}\}.
\end{equation}
CPTP represents the set of all completely positive channels on $\mathcal{D}(\mathcal{H})$. Any set of free operations is a subset of RNG. Different sets of free operations lead to different resource theories while the resource theories with RNG might have a universal property~\cite{Brand2015Reversible}.

After defining the three ingredients in QRT, we need to characterize the resource by providing a resource measure. The measure of the resource is a functional $M$ mapping from $\mathcal{D}(\mathcal{H})$ to non-negative real numbers. A valid measure $M$ should satisfy the following two conditions. First, it is 0 for the set of free states:
\begin{equation}\label{eq:R1}
	M(\sigma) = 0, \forall \sigma\in \mathcal{F}.
\end{equation}
In some QRTs this requirement is more strict. $M$ is 0 if and only if the state is free state:
\begin{equation}\label{eq:R1'}
	M(\sigma) = 0 \Leftrightarrow \sigma\in \mathcal{F}.
\end{equation}
Second, any proper resource measure $M$ cannot increase under the action of free operations. This is the monotone condition. Then, for any free operation $\Lambda$,
\begin{equation}\label{eq:R2}
	M(\Lambda(\rho))\leq M(\rho), \forall \rho.
\end{equation}
There are some additional requirements for the measure in different QRTs. These extra requirements could vary a lot for different QRTs. In the resource theory of coherence, it is reasonable that coherence cannot increase under mixing from a physical point of view. This leads to the convexity condition of the measure:
\begin{equation}\label{eq:convex}
M\left(\sum\limits_n c_n \rho_n\right)\leq \sum\limits_n c_nM(\rho_n),\sum\limits_n c_n = 1.
\end{equation}

Here, we introduce two kinds of measures. The first is divergence measure, that is, the measure of the resource bases on the divergence between the state and the set of free states $\mathcal{F}$. Divergence is a functional mapping two quantum states into a non-negative real number:
\begin{equation}
	D:\mathcal{D}(\mathcal{H})\times \mathcal{D}(\mathcal{H})\rightarrow \mathbb{R}^{+},
\end{equation}
requiring $D(\rho,\sigma) = 0\Leftrightarrow \rho = \sigma$. The divergence of a state $\rho$ and the set of free states $\mathcal{F}$ is defined as
\begin{equation}
D(\rho,\mathcal{F}) = \inf_{\sigma \in \mathcal{F}} D(\rho,\sigma).
\end{equation}
To get a well-defined divergence measure, we require $\mathcal{F}$ to be a convex set, i.e., for any $\rho,\sigma\in\mathcal{F}$, $t\rho+(1-t)\sigma\in\mathcal{F},\forall 0\leq t\leq 1$. Normally, we take $D$ as the K-L divergence or relative entropy $S$: $S(\rho||\sigma) = \tr(\rho\log \sigma)-\tr(\sigma\log \sigma)$. The relative entropy of the resource is $M(\rho) = S(\rho||\mathcal{F})$. We can prove that this measure meets the requirements of Eq.~\eqref{eq:R1'}, Eq.~\eqref{eq:R2} and Eq.~\eqref{eq:convex}.

It can be verified that divergence measure satisfies Eq.~\eqref{eq:R1'}. Due to the contractive property of relative entropy $S$, i.e., for any CPTP channel $\mathcal{E}$,
\begin{equation}
	S(\mathcal{E}(\rho)||\mathcal{E}(\sigma))\leq S(\rho||\sigma),
\end{equation}
the monotone condition Eq.~\eqref{eq:R2} can be fulfilled:
\begin{equation}
	\begin{split}
		M(\Lambda(\rho)) &= \inf_{\sigma\in \mathcal{F}}S(\Lambda(\rho)||\sigma)\\
		 &\leq \inf_{\sigma\in \mathcal{F}}S(\Lambda(\rho)||\Lambda(\sigma))\\
		&\leq \inf_{\sigma\in \mathcal{F}}S(\rho||\sigma)\\ &= M(\rho).
	\end{split}
\end{equation}
Moreover, relative entropy $S$ is jointly convex, then
\begin{equation}\label{eq:jointconvex}
	\begin{split}
		M\left(\sum\limits_n c_n\rho_n\right)&\leq S\left(\sum\limits_n c_n\rho_n || \sum\limits_n c_n\sigma^*_n\right)\\
		&\leq \sum\limits_n c_n S(\rho_n || \sigma^*_n)\\
		&= \sum\limits_n c_n M(\rho_n),
	\end{split}
\end{equation}
where $\sigma^*_n$ is the closest quantum state to $\rho_n$ in $\mathcal{F}$. Equation.~\eqref{eq:jointconvex} leads to the convexity of the measure.

Another widely used resource measure is robustness of the resource, that is how hard to make the state become a free state with mixing another state. For any state $\rho\in\mathcal{D}(\mathcal{H})$, the robustness of the resource is
\begin{equation}\label{eq:robustnessdef}
\begin{split}
R(\rho) &= \min_{\sigma\in \mathcal{D}(\mathcal{H})} s,\\ s.t.\ \frac{\rho+s\sigma}{1+s}&\in \mathcal{F}, s\geq 0.
\end{split}
\end{equation}
Obviously robustness measure satisfies Eq.~\eqref{eq:R1'}.
We can prove it also satisfies monotone condition. For any quantum state $\rho$, we set $r = R(\rho)$. According to the definition of robustness measure, $\exists \sigma\in\mathcal{D}(\mathcal{H})$, $\frac{\rho+r\sigma}{1+r}\in\mathcal{F}$. Then for any free operation $\Lambda$,
\begin{equation}
\Lambda\left(\frac{\rho+r\sigma}{1+r}\right) = \frac{\Lambda(\rho)+r\Lambda(\sigma)}{1+r} \in \mathcal{F}.
\end{equation}
From Eq.~\eqref{eq:robustnessdef} we can see $R(\Lambda(\rho))\leq r = R(\rho)$. That is the condition of Eq.~\eqref{eq:R2}.

Moreover, the robustness measure meets the requirement of Eq.~\eqref{eq:convex} if the set of free states $\mathcal{F}$ is convex. Denote $r_i = R(\rho_i)$ for a set of states $\{\rho_1, \rho_2, \cdots, \rho_n\}$, then $\exists \{\sigma_1, \sigma_2, \cdots, \sigma_n\}$, $\frac{\rho_i+r_i\sigma}{1+r_i}\in\mathcal{F}, \forall i$. For any convex combination of $\{\rho_1, \rho_2, \cdots, \rho_n\}$: $\rho = \sum_i c_i\rho_i$, $\sum_i c_i = 1, c_i\geq 0$, set $r = \sum_i c_i r_i$, $\sigma = \frac{\sum_i c_i r_i \sigma_i}{r}$,
\begin{equation}
\frac{\rho + r\sigma}{1+r} = \frac{\sum_i c_i(\rho_i+r_i\sigma_i)}{1+r} \in \mathcal{F}.
\end{equation}
This accounts for the convexity of the robustness measure.

\subsection{Resource Theory of Coherence}
Here we apply the framework of resource theory to characterize coherence. In the resource theory of coherence~\cite{Baumgratz2014Coherence}, the set of free states and free operations are named incoherent states and incoherent operations. We first set the computational basis of Hilbert space $\mathcal{H}$, denoted as $\{\ket{i},i\in[d]\}$. Then we define the incoherent states to be the states only with diagonal terms in computational basis, i.e.,
\begin{equation}\label{eq:incoherentstate}
\mathcal{I} = \left\{\sum_i c_i\ketbra{i}\big| \sum_i c_i = 1, c_i\geq 0\right\}.
\end{equation}
The set of incoherent operations has different choices. The largest set of incoherent operations (RNG) is called maximally incoherent operations (MIO) in the resource theory of coherence. Any operation in MIO cannot generates coherence from incoherent states. It can be proved that~\cite{LiuZiWen2017ResourceDestroying}
\begin{equation}\label{eq:MIO}
\mathrm{MIO} = \{O\in\mathrm{CPTP}|O\circ \Delta = \Delta\circ O\circ \Delta\},
\end{equation}
where $\Delta$ is dephasing operation on $\mathcal{D}(\mathcal{H})$. A smaller set of incoherent operations is dephasing-covariant incoherent operations (DIO):
\begin{equation}
	\mathrm{DIO} = \{O\in\mathrm{CPTP}|O\circ \Delta = O\circ \Delta\}.
\end{equation}
DIO is the set of operations commuting with the dephasing operation. Dephasing operation can be viewed as the resource destroying map~\cite{LiuZiWen2017ResourceDestroying} in the resource theory of coherence. From it we can define another set of operations named coherence non-activating operations (CNAO):
\begin{equation}\label{eq:CNAO}
\mathrm{CNAO} = \{ O\in\mathrm{CPTP}|\Delta \circ O = \Delta \circ O\circ \Delta \}.
\end{equation}
CNAO is not a subset of MIO so it cannot be specified as the set of free operations.

Provided the incoherent states and incoherent operations, we need to find coherence measure. We can verify that the set of incoherent states is convex. Then applying the conclusion in Subsection~\ref{appendssc:resourceframework}, we define the relative entropy of coherence and robustness of coherence. The relative entropy of coherence $C_{rel}$ is
\begin{equation}
\begin{split}
C_{rel}(\rho) &= S(\rho||\mathcal{I})\\
&= \inf_{\sigma \in \mathcal{I}} S(\rho||\sigma)\\
&= S(\rho||\Delta(\rho))\\
&= S(\Delta(\rho)) - S(\rho),
\end{split}
\end{equation}
where $S(\rho) = -\tr(\rho\log \rho)$ is the von-neumann entropy. The robustness of coherence $R$ is
\begin{equation}
\begin{split}
R(\rho) &= \min_{\sigma\in \mathcal{D}(\mathcal{H})} s,\\ s.t.\ \frac{\rho+s\sigma}{1+s}&\in \mathcal{I}, s\geq 0.
\end{split}
\end{equation}
With these two measures we can study the interconversion of states under different sets of incoherent operations and explore the quantum advantage shown in coherence.

\subsection{Channel Resource Theory}
Previous two subsections discuss the state resource theory. We can also establish a resource theory framework for the channels. Channel is the completely positive and trace preserving map on $\mathcal{D}(\mathcal{H})$. In a channel resource theory, the resource is a functional mapping a channel instead of a state to the non-negative number. We need to specify the set of free channels and the set of free superchannels as the analogy of the set of free states and the set of free operations. The superchannel transforms one channel to another just like the operation acting on the states. We can choose RNG superchannels as the free superchannels when we establish a channel resource theory to avoid providing a specific representation of superchannels.

Similar with the state resource theory, we need to provide the resource measure. One approach to discussing resource measure is transforming the channel to its Choi-state representation, then characterize the resource of Choi-state with the tool of state resource theory. Here we briefly introduce this approach. The Choi-state representation of a channel $\mathcal{N}$ is
\begin{equation}
	\begin{split}
		\Phi_{\mathcal{N}} &= I\otimes \mathcal{N} (\ketbra{\Phi^+})\\
		&= \frac{1}{d}I\otimes \mathcal{N} \sum_{ij}\ketbra{ii}{jj}\\
		&= \frac{1}{d} \sum_{ij}\ketbra{i}{j}\mathcal{N}(\ketbra{i}{j})\\
		&= \frac{1}{d} \sum_{i,k,j,l} \mathcal{N}_{ik,jl}\ketbra{i}{j}\otimes \ketbra{k}{l}\\
		&= \frac{1}{d} \sum_{i,k,j,l} \mathcal{N}_{ik,jl}\ketbra{ik}{jl},
	\end{split}
\end{equation}
where $\ket{\Phi^+}$ is the maximally entangled state $\frac{1}{\sqrt{d}}\sum_{i}\ket{ii}$ on $\mathcal{H}\otimes \mathcal{H}$, $\mathcal{N}_{ik,jl} = \tr(\ketbra{l}{k}\mathcal{N}(\ketbra{i}{j}))$. Then the set of free channels can be transformed into a set of free states. The set of free superchannels will become a set of free operations. As long as the set of free channels is convex, the corresponding set of Choi-states is convex. After that we can define the divergence measure and robustness measure in the channel resource theory.

A concrete example of channel resource theory is entanglement breaking (EB) channel~\cite{Horodecki2003EBChannel}. A channel $\mathcal{N}$ is entanglement breaking if and only if its Choi-state $\Phi_{\mathcal{N}}$ is separable. At the same time, any EB channel has the form:
\begin{equation}
\mathcal{N} = \sum_i \sigma_i \tr(M_i \rho),
\end{equation}
where $\sigma_i$ is a density operator and $\{M_i\}$ form a POVM. That is, any EB channel can be realized by measuring the state with a POVM, and then preparing a state based on the measurement result. The set of EB channels is convex. We can take this set to be free channels and quantify the resource of a channel relative to this set.

\section{Partial Replacement of Quantum Operation}\label{append:PartReplace}
Here, we discuss the scenario of replacing part of a quantum operation with classical processing in detail. For simplicity, we only consider the case of qcCRO. First, we focus on the case where the evolution of a subsystem originating from a large evolution on the overall system can be viewed as a CPTP map.

Consider a system composed of two subsystems, $\mathcal{H}_a$ and $\mathcal{H}_b$, as shown in Figure~\ref{fig:ReplaceDestroy}, where the quantum state of the whole system, $\mathcal{H}_{ab} \equiv \mathcal{H}_a\otimes \mathcal{H}_b$, is a product state. The quantum evolution on the whole system is given by a unitary operation $U$. If we focus on the subsystem $\mathcal{H}_b$ and ignore subsystem $\mathcal{H}_a$, the equivalent action on $\mathcal{H}_b$ is given by
\begin{equation}
O(\rho_b) = \tr_a(U\rho_a\otimes \rho_b U^{\dagger}).
\end{equation}
If we know the initial state of $\rho_a$ and the unitary $U$, we can obtain the form of $O$. Here, we set specific values for $\rho_a$ and $U$ so that $O$ is a qcCRO. Then, $O$ can be moved after the computational basis measurement and replaced by classical processing $O_c$.

\begin{figure}[!htbp]
	\centering \resizebox{15cm}{!}{\includegraphics{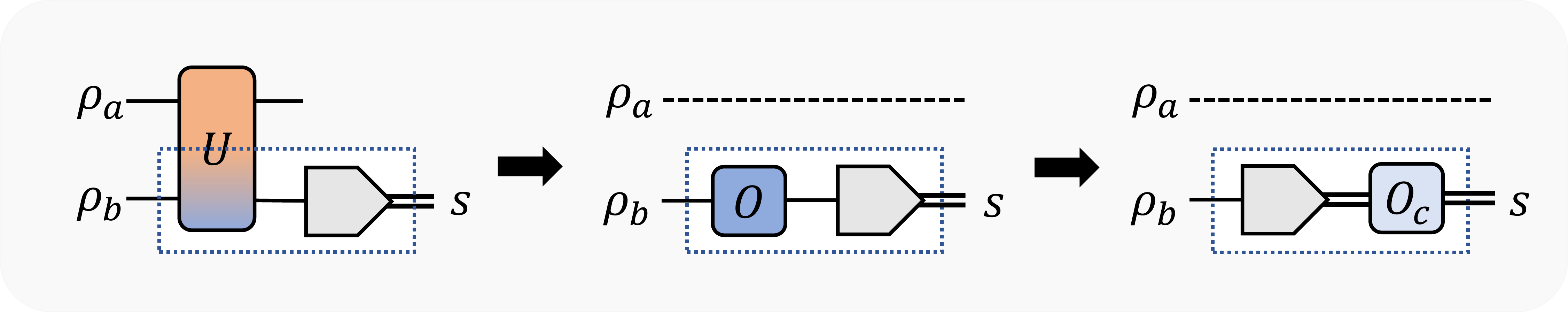}}
	\caption{Classical replacement for part of a quantum operation. The whole system $\mathcal{H}_{ab}$ is composed of two subsystems $\mathcal{H}_a$ and $\mathcal{H}_b$. The unitary $U$ is an evolution on the whole system and it becomes a qcCRO, $O$, when focusing on the evolution of the state on subsystem $\mathcal{H}_b$. The qcCRO $O$ can be moved after measurement and be replaced by classical processing $O_c$. The dashed lines in the second and third figures mean that we do not consider the evolution of $\rho_a$. In the first figure, there exists a correlation between two subsystems in the final state, while in the third figure, the correlation is destroyed.}
	\label{fig:ReplaceDestroy}
\end{figure}

Before classical replacement, the final state on system $\mathcal{H}_{ab}$ is
\begin{equation}
\rho_{ab}' = \sum_i \tr_b[U(\rho_a\otimes \rho_b)U^{\dagger} \ketbra{i}] \otimes \ketbra{i},
\end{equation}
which exhibits correlation between subsystem $\mathcal{H}_a$ and $\mathcal{H}_b$. After classical replacement, the interaction between $\mathcal{H}_a$ and $\mathcal{H}_b$ is eliminated, so there is no correlation between two subsystems in the end. In this sense, we conclude that replacing a CRO only guarantees the subsystem it acts on is unchanged while the other systems and the correlation between the subsystem and the rest might change.

In Figure~\ref{fig:ReplaceAPart}, we show an example replacing a CRO while maintaining the state of the whole system after classical replacement. The two subsystems $\mathcal{H}_a$ and $\mathcal{H}_b$ are both qubits with computational basis as $Z$ basis. The initial state on the whole system is $\rho_a\otimes \rho_b$ and the interaction is a control unitary from $\mathcal{H}_a$ to $\mathcal{H}_b$. The action on $\mathcal{H}_a$ after ignoring $\mathcal{H}_b$ is a dephasing channel, for any state $\rho_a\in \mathcal{D}(\mathcal{H}_a)$,
\begin{equation}
O(\rho_a) = (1-p)\rho_a + pZ\rho_a Z,
\end{equation}
where $p$ is dephasing rate depending on the interaction and initial state on $\mathcal{H}_b$. Here, $O$ is a qcCRO and can be moved after computational basis measurement. The corresponding classical processing $O_c$ is identity, mapping a classical bit to itself. To ensure the state on $\mathcal{H}_{ab}$ does not change before and after classical replacement, we add a classical-quantum control from subsystem $a$ to subsystem $b$. In this way, not only the subsystem the CRO acts on is unchanged, the whole system is unchanged as well.

\begin{figure}[!htbp]
	\centering \resizebox{9cm}{!}{\includegraphics{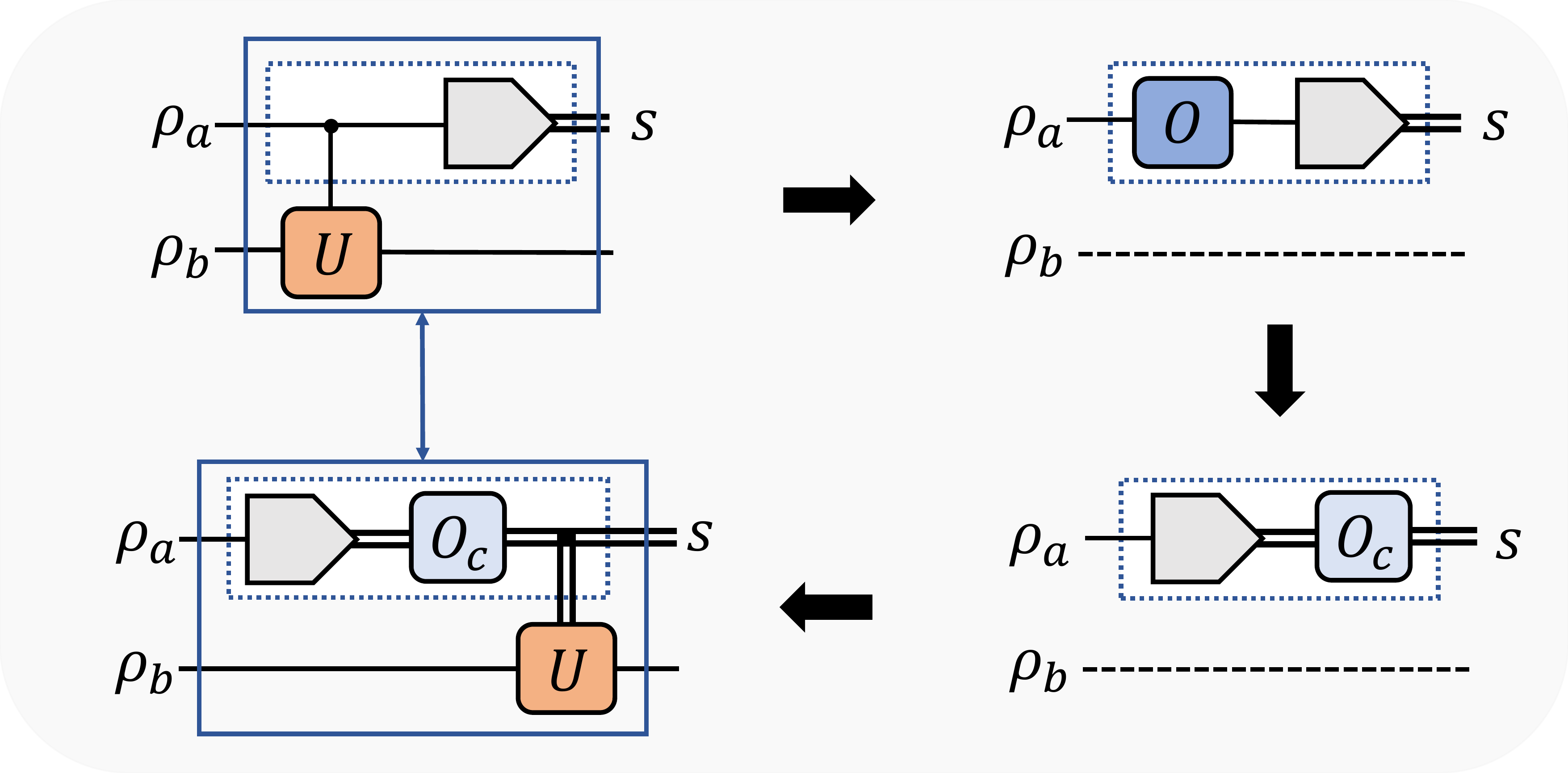}}
	\caption{Given quantum system $\mathcal{H}_a\otimes \mathcal{H}_b$ with an initial state, $\rho_a\otimes \rho_b$, control-$U$ from subsystem $a$ to subsystem $b$ would become a dephasing operation, $O$, when focusing on subsystem $a$. The dephasing operation, $O$, is a qcCRO whose form depends on $\rho_b$, while its corresponding classical processing $O_c$ is irrelevant to $\rho_b$, means that we can replace $O$ even if we do not know its concrete form. After adding classical-quantum control-$U$ to the circuit after replacement, the final state of the whole system is unchanged before and after replacement.}
	\label{fig:ReplaceAPart}
\end{figure}

It is worth noting that Figure~\ref{fig:ReplaceAPart} shows an example where we can perform classical replacement without knowing the concrete form of a CRO. In this example, the dephasing rate $p$ is unknown if the initial state on subsystem $\mathcal{H}_b$ is uncharacterized. Nevertheless, for any value of $p$, $O$ is related to a fixed classical processing so that the replacement can be implemented.

In the discussion above, we restrict the initial state to be a product state so that the evolution on the whole system reduces to a CPTP map when we look at the subsystem. However, if the initial state is entangled, the evolution on a subsystem may not be able to be expressed as a CPTP map, as shown in Figure~\ref{fig:NonMarkovianReplace1}. The classical replacement in this case is different from replacing a \textit{quantum operation}, which is another interesting problem.

\begin{figure}[!htbp]
	\centering \resizebox{10cm}{!}{\includegraphics{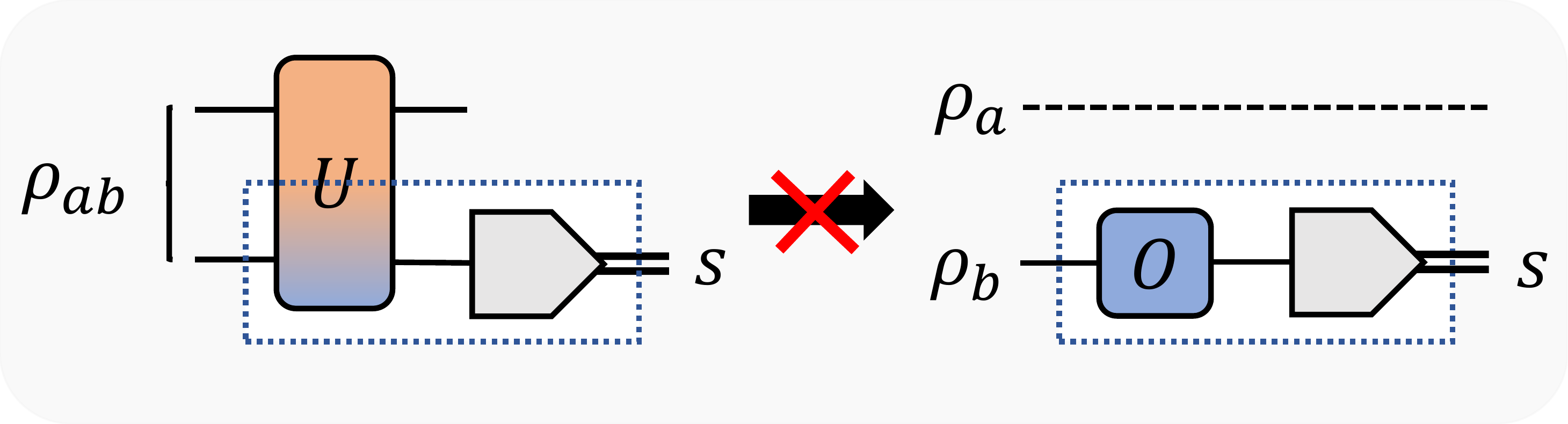}}
	\caption{Given quantum system $\mathcal{H}_{a}\otimes \mathcal{H}_b$ with an entangled initial state, $\rho_{ab}$, the interaction between two subsystems is given by $U$. Due to the quantum correlation of the initial state, $U$ cannot reduce to a CPTP map when focusing on subsystem $b$.}
	\label{fig:NonMarkovianReplace1}
\end{figure}

The situation where a quantum evolution cannot reduce to a CPTP map on a subsystem comes from the violation of the Markovian assumption. Another example of non-Markovian evolution on the subsystem is shown in Figure~\ref{fig:NonMarkovianReplace2}. With a product initial state $\rho_a\otimes \rho_b$ on system $\mathcal{H}_a\otimes \mathcal{H}_b$, two unitary evolutions on the whole system, $U_1$ and $U_2$, can reduce to two qcCROs on subsystem $b$, $O_1$ and $O_2$, respectively. However, the concatenation of two unitary evolutions, $U_2\circ U_1$, cannot reduce to $O_2\circ O_1$. An interesting question for future research is that when $U_2\circ U_1$ can reduce to a qcCRO in such a case. A positive example is when $U_1$ and $U_2$ are both CNOT gates. Here, $U_2\circ U_1$ cannot reduce to $O_2\circ O_1$ but can reduce to an identity operation, which is a qcCRO.

\begin{figure}[!htbp]
	\centering \resizebox{10cm}{!}{\includegraphics{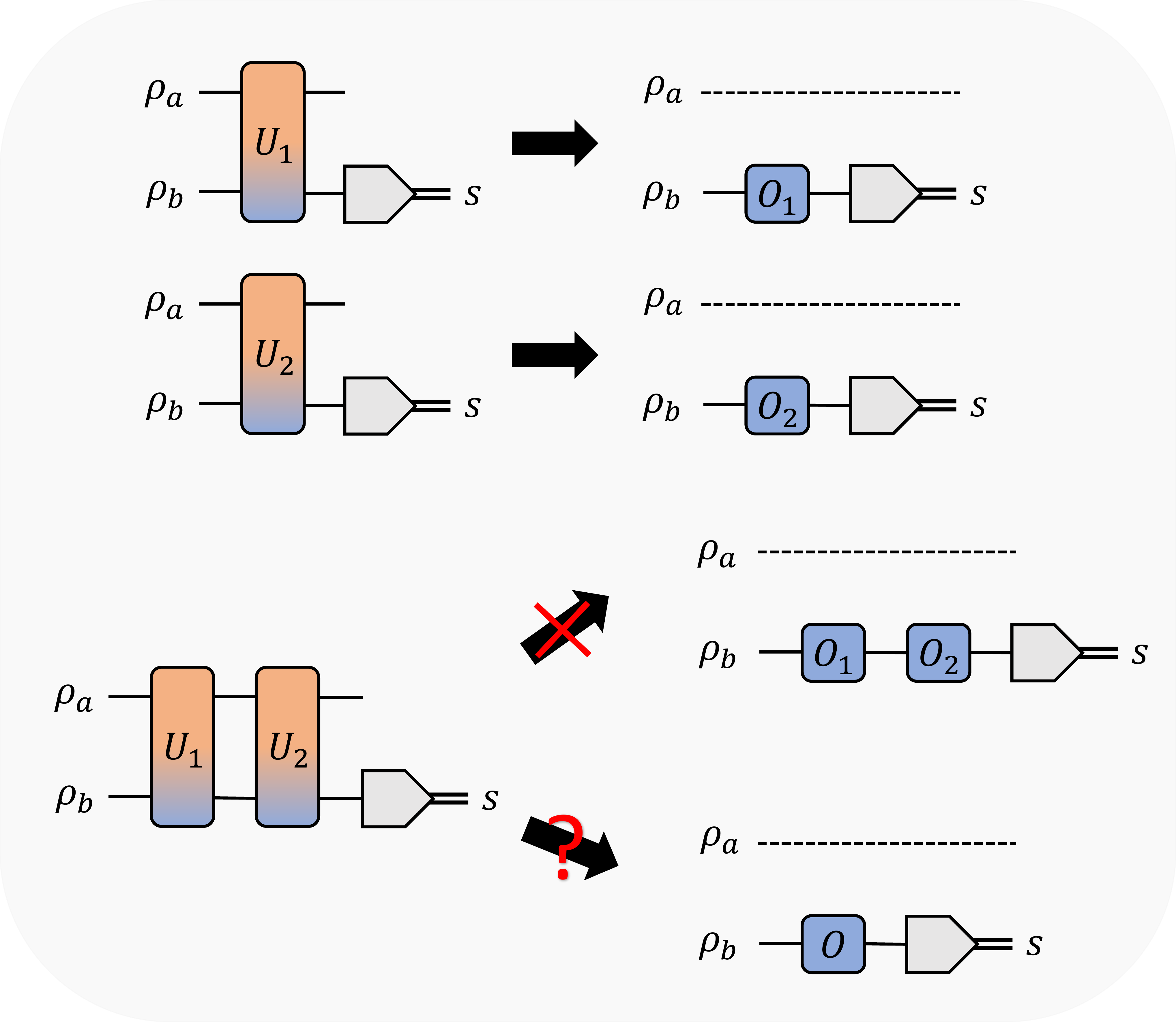}}
	\caption{Given quantum system $\mathcal{H}_a\otimes \mathcal{H}_b$ with a product initial state, $\rho_a\otimes \rho_b$, the unitary evolutions, $U_1$ and $U_2$, can both reduce to qcCROs $O_1$ and $O_2$ on subsystem $b$ respectively. The concatenation of two unitary evolutions, $U_2U_1$, cannot reduce to $O_2\circ O_1$, as in general the evolution on subsystem $b$ is non-Markovian. Interestingly, in some cases, $U_2U_1$ can reduce to a qcCRO, $O$, while the general case remains an open problem.}
	\label{fig:NonMarkovianReplace2}
\end{figure}

\section{Proof of Theorem~\ref{thm:CROequiv} and Theorem~\ref{thm:CROextequiv}}\label{thmproof:CROequiv}
Here, we provide a proof of Theorem~\ref{thm:CROequiv} and Theorem~\ref{thm:CROextequiv}. The former is a special case of the latter so we only prove Theorem~\ref{thm:CROextequiv}. We prove two lemmas to verify the equivalence of the second and the third criteria, and the equivalence of the first and the second, respectively, for each kind of CRO. The notations are the same in the main-text.
\subsection{Proof of the Equivalence of the Second and the Third Criteria}
We start with proving Lemma~\ref{lemma:qcCRO} as shown below.
\begin{lemma}\label{lemma:qcCRO}
Given a quantum operation, $O$, and a quantum channel, $\mathcal{T}_E(\rho) = \sum_n \frac{E_n}{\tr(E_n)}\tr(\rho E_n)$, defined with a set of projectors $\{E_n\}$, then
\begin{itemize}
\item
$\exists O^{\prime}\in\mathrm{CPTP}, O\circ \mathcal{T}_E = \mathcal{T}_E\circ O^{\prime}\circ \mathcal{T}_E\Rightarrow O\circ \mathcal{T}_E = \mathcal{T}_E\circ O\circ \mathcal{T}_E$;

\item
$\exists O^{\prime}\in\mathrm{CPTP}, O = \mathcal{T}_E\circ O^{\prime}\circ \mathcal{T}_E\Rightarrow O = \mathcal{T}_E\circ O\circ \mathcal{T}_E$;

\item
$\exists O^{\prime}\in\mathrm{CPTP}, \mathcal{T}_E\circ O = \mathcal{T}_E\circ O^{\prime}\circ \mathcal{T}_E\Rightarrow \mathcal{T}_E\circ O = \mathcal{T}_E\circ O\circ \mathcal{T}_E$.
\end{itemize}
\end{lemma}
\begin{proof}
If operation $O$ satisfies $\exists O^{\prime}\in\mathrm{CPTP}, O\circ \mathcal{T}_E = \mathcal{T}_E\circ O^{\prime}\circ \mathcal{T}_E$, then,
\begin{equation}
\begin{split}
\mathcal{T}_E\circ O \circ \mathcal{T}_E &= \mathcal{T}_E\circ \mathcal{T}_E \circ O^{\prime} \circ \mathcal{T}_E\\
&= \mathcal{T}_E\circ O^{\prime}\circ \mathcal{T}_E\\
&= O\circ \mathcal{T}_E.
\end{split}
\end{equation}
Here, the second equality comes from the idempotence of $\mathcal{T}_E$,
\begin{equation}
\begin{split}
\mathcal{T}_E\circ \mathcal{T}_E (\rho) &= \sum_m\sum_n \frac{E_m}{\tr(E_m)} \frac{\tr(E_mE_n)}{\tr(E_n)} \tr(\rho E_n)\\
&= \sum_m\sum_n \frac{E_m}{\tr(E_m)} \frac{\delta_{nm}\tr(E_n)}{\tr(E_n)} \tr(\rho E_n)\\
&= \sum_n \frac{E_n}{\tr(E_n)} \tr(\rho E_n)\\
&= \mathcal{T}_E(\rho).
\end{split}
\end{equation}
Similarly, if operation $O$ satisfies $\exists O^{\prime}\in\mathrm{CPTP}, O = \mathcal{T}_E\circ O^{\prime}\circ \mathcal{T}_E$, then,
\begin{equation}
	\begin{split}
		\mathcal{T}_E\circ O \circ \mathcal{T}_E &= \mathcal{T}_E\circ \mathcal{T}_E\circ O^{\prime}\circ \mathcal{T}_E \circ \mathcal{T}_E\\
		&= \mathcal{T}_E\circ O^{\prime}\circ \mathcal{T}_E\\
		&= O.
	\end{split}
\end{equation}
If operation $O$ satisfies $\exists O^{\prime}\in\mathrm{CPTP}, O\circ \mathcal{T}_E = \mathcal{T}_E\circ O^{\prime}\circ \mathcal{T}_E$, then,
\begin{equation}
	\begin{split}
		\mathcal{T}_E\circ O \circ \mathcal{T}_E &= \mathcal{T}_E\circ O^{\prime}\circ \mathcal{T}_E \circ \mathcal{T}_E\\
		&= \mathcal{T}_E\circ O^{\prime}\circ \mathcal{T}_E\\
		&= \mathcal{T}_E\circ O.
	\end{split}
\end{equation}
\end{proof}
From Lemma~\ref{lemma:qcCRO} we obtain that the third criterion implies the second for $\mathrm{cqCRO}_{\{E_n\}}$, $\mathrm{qqCRO}_{\{E_n\}}$, and $\mathrm{qcCRO}_{\{E_n\}}$ respectively. The opposite is also true as the second criterion means that there exists $O' = O$, $O\circ \mathcal{T}_E= \mathcal{T}_E\circ O^{\prime}\circ \mathcal{T}_E$, $O = \mathcal{T}_E\circ O^{\prime}\circ \mathcal{T}_E$, $\mathcal{T}_E\circ O = \mathcal{T}_E\circ O^{\prime}\circ \mathcal{T}_E$, respectively. As a consequence, the two criteria are equivalent for $\mathrm{cqCRO}_{\{E_n\}}$, $\mathrm{qqCRO}_{\{E_n\}}$, and $\mathrm{qcCRO}_{\{E_n\}}$ respectively.

\subsection{Proof of the Equivalence of the First and the Second Criteria}
For the next proof, we first prove the equivalence of the first and second criteria for $\mathrm{qcCRO}_{\{E_n\}}$, and then discuss the cases of other two CROs. We begin with defining an extension of qcCRO and $\mathrm{qcCRO}_{\{E_n\}}$ and propose Lemma~\ref{lemma:qcCROext2}.

Given any input state $\rho\in \mathcal{D}(H)$, any quantum channel $O$, we perform the PVM, $M_1 = \{E_n\}$, on $O(\rho)$ to get the measurement result $s$. We call $O$ a $\mathrm{qcCRO}_{\{E_n\},\{F_m\}}$ if we can get the same result for any same input by measuring $\rho$ with PVM, $M_2 = \{F_m\}$ followed with a classical processing $O_c$ as shown in Figure~\ref{fig:CROproof}. The two measurements $M_1$ and $M_2$ are predetermined and we set $\mathcal{T}_E(\rho) = \sum_n \frac{E_n}{\tr(E_n)}\tr(\rho E_n)$ and $\mathcal{T}_F(\rho) = \sum_m \frac{F_m}{\tr(F_m)}\tr(\rho F_m)$.

\begin{figure}[!htbp]
	\centering \resizebox{6cm}{!}{\includegraphics{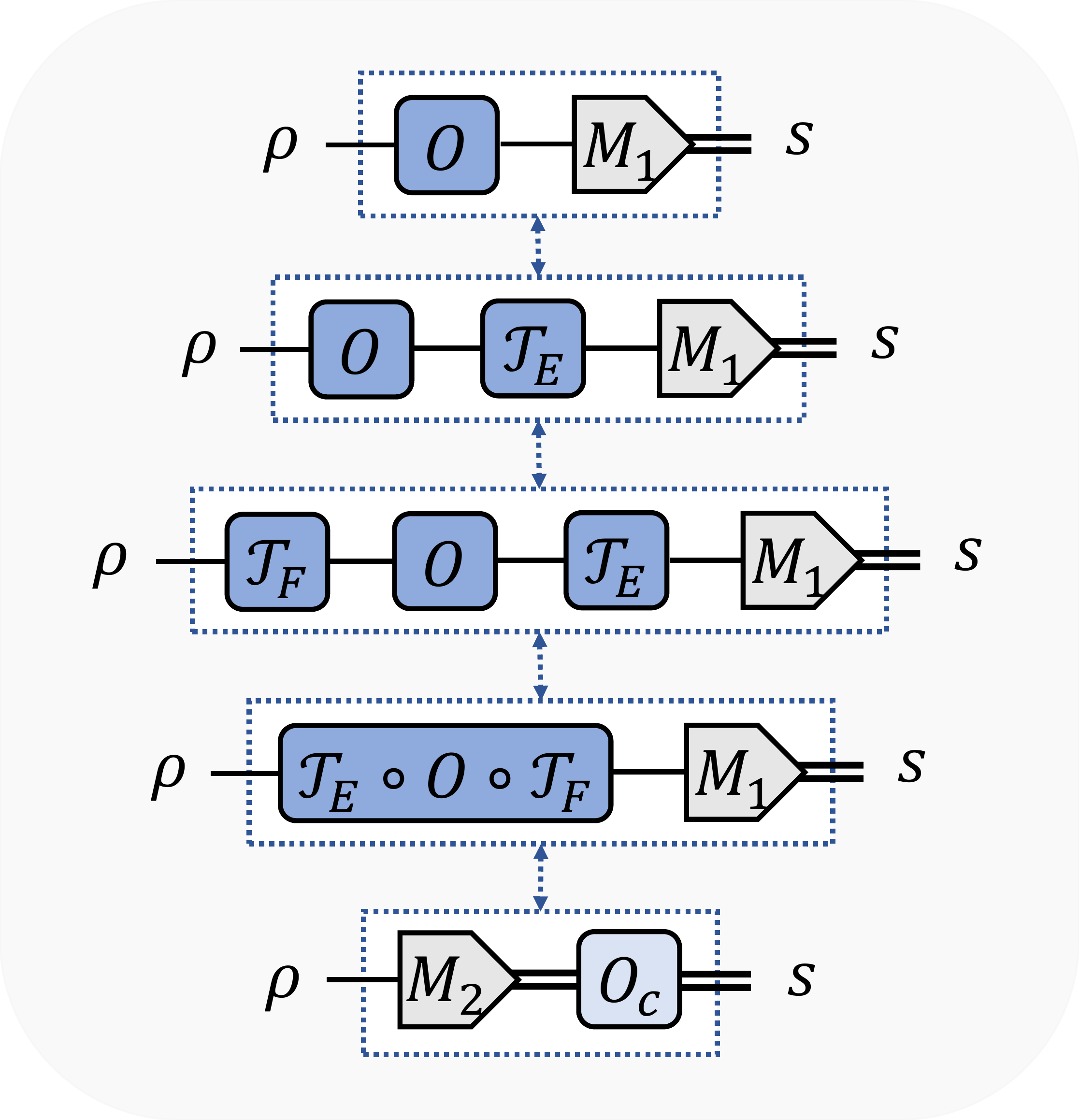}}
	\caption{Given any input state $\rho\in \mathcal{D}(H)$, for any operation, $O$, in Eq.~\eqref{eq:qcCROext2}, the measurement result of PVM, $M_1$, on $O(\rho)$ is equivalent to the measurement result of PVM, $M_2$, on $\rho$ followed by classical processing $O_c$.}
	\label{fig:CROproof}
\end{figure}

In general, the two measurements $M_1,M_2$ are different and $\mathcal{T}_F$ does not commute with $\mathcal{T}_E$. The case of $\mathrm{qcCRO}_{\{E_n\},\{F_m\}}$ reduces to $\mathrm{qcCRO}_{\{E_n\}}$ when $M_1=M_2$, and further reduces to qcCRO if $M_1$ and $M_2$ are both computational basis measurement. Similar with qcCRO and $\mathrm{qcCRO}_{\{E_n\}}$, we provide a necessary and sufficient criterion for $\mathrm{qcCRO}_{\{E_n\},\{F_m\}}$ as shown in Lemma~\ref{lemma:qcCROext2}. Then we can express the set of $\mathrm{qcCRO}_{\{E_n\},\{F_m\}}$ as
\begin{equation}\label{eq:qcCROext2}
	\mathrm{qcCRO}_{\{E_n\},\{F_m\}} = \{O\in\mathrm{CPTP}|\mathcal{T}_E\circ O = \mathcal{T}_E\circ O\circ \mathcal{T}_F\}.
\end{equation}
\begin{lemma}\label{lemma:qcCROext2}
For any $\mathrm{CPTP}$ map $O$, $O\in\mathrm{qcCRO}_{\{E_n\},\{F_m\}}\Leftrightarrow \mathcal{T}_E\circ O = \mathcal{T}_E\circ O\circ \mathcal{T}_F$.
\end{lemma}

\begin{proof}
We first find the corresponding classical processing for each operation satisfying Eq.~\eqref{eq:qcCROext2}, and then prove that any operation outside Eq.~\eqref{eq:qcCROext2} cannot.

For any operation $O$ in Eq.~\eqref{eq:qcCROext2}, the measurement result of $M_1$ on $O(\rho)$ is equivalent to the measurement result of $\mathcal{T}_E \circ  O\circ \mathcal{T}_F (\rho)$ as shown in Figure~\ref{fig:CROproof}. Here, we utilise the measurement result of $M_1$ does not change if inserting $\mathcal{T}_E$ before measurement and $\mathcal{T}_E \circ O = \mathcal{T}_E \circ O \circ \mathcal{T}_F$. As
\begin{equation}
	\mathcal{T}_E \circ  O\circ \mathcal{T}_F (\rho) = \sum_{n,m}  \tr(F_m\rho) \tr(E_nO(\frac{F_m}{\tr(F_m)})) \frac{E_n}{\tr(E_n)},
\end{equation}
the measurement result $s$ of $M_1$ on $O(\rho)$ takes the value $n$ with probability $\sum_m \tr(F_m\rho) \tr(E_n O(\frac{F_m}{\tr(F_m)}))$. It is equivalent to measure $\rho$ with $M_2$ to get $m$ with probability $\tr(F_m \rho)$, and then transfer $m$ to $n$ with probability $\tr(E_n O(\frac{F_m}{\tr(F_m)}))$. We can view $\tr(F_m \rho)$ as an initial probability distribution and $\tr(E_n O(\frac{F_m}{\tr(F_m)}))$ as a stochastic matrix of the classical processing. Note that the term $\tr(E_n O(\frac{F_m}{\tr(F_m)}))$ is irrelevant to the initial state $\rho$ so the classical replacement can be realized. Then we find the corresponding classical processing for any operation satisfying Eq.~\eqref{eq:qcCROext2} and complete the proof of the first step.

For the second step, we show that for any operation $O$ outside Eq.~\eqref{eq:qcCROext2}, there exists two different states, $\rho_1$ and $\rho_2$ satisfying $\mathcal{T}_F (\rho_1) = \mathcal{T}_F (\rho_2)$, while $M_1$ on $O(\rho_1)$ and $O(\rho_2)$ can output different probability distributions. Note that $\mathcal{T}_F (\rho_1) = \mathcal{T}_F (\rho_2)$ implies PVM $M_2$ cannot distinguish the classical measurement results of the two. Then the classical processing only outputs the same results for these two inputs. That means for either $\rho_1$ or $\rho_2$ the quantum operation and classical processing cannot output the same result, then we prove $O$ cannot be replaced by classical processing.

If an operation $O$ does not satisfy Eq.~\eqref{eq:qcCROext2}, then $\mathcal{T}_E \circ O \neq \mathcal{T}_E \circ O \circ \mathcal{T}_F$. That means there exists a state $\sigma$ satisfying
\begin{equation}
\mathcal{T}_E \circ O(\sigma) \neq \mathcal{T}_E \circ O ( \mathcal{T}_F(\sigma) ).
\end{equation}
Now we take $\rho_1 = \sigma$, $\rho_2 = \mathcal{T}_F(\sigma)$. Obviously, $\mathcal{T}_F (\rho_1) = \mathcal{T}_F (\rho_2)$. As $\mathcal{T}_E (O(\rho_1)) \neq \mathcal{T}_E (O(\rho_2))$, the measurement results of $M_1$ on $O(\rho_1)$ and $O(\rho_2)$ are different. Then we complete the second step of the proof.
\end{proof}
The equivalence of the first and the second criteria for qcCRO and $\mathrm{qcCRO}_{\{E_n\}}$ is proved as a special case of Lemma~\ref{lemma:qcCROext2}.

For $\mathrm{cqCRO}_{\{E_n\}}$, we turn back to the scenario where the state preparation in the classical replacement is the same as that right before the replaced quantum operation. They are both state preparation associated with projectors $\{E_n\}$. For any operation $O$ satisfying $O\circ \mathcal{T}_E = \mathcal{T}_E \circ O \circ \mathcal{T}_E$, and classical input $n$, the output state is $\sum_{m} \frac{E_m}{\tr(E_m)} \tr(E_mO(\frac{E_n}{\tr(E_n)}))$. That is first transforming $n$ into $m$ with probability $\tr(E_mO(\frac{E_n}{\tr(E_n)}))$ classically, and then preparing the state. Reversely, if $O$ does not satisfy $O\circ \mathcal{T}_E = \mathcal{T}_E \circ O \circ \mathcal{T}_E$, there exists a quantum state $\sigma$ satisfying
\begin{equation}\label{eq:cqCROcounter}
O\circ \mathcal{T}_E(\sigma) \neq \mathcal{T}_E \circ O \circ \mathcal{T}_E(\sigma).
\end{equation}
Set $\sigma_E = \mathcal{T}_E(\sigma) = \sum_n \frac{E_n}{\tr(E_n)}\tr(\sigma E_n)$, then we can see that $\sigma_E$ can be obtained from state preparation by inputting $n$ with probability $\tr(\sigma E_n)$. From Eq.~\eqref{eq:cqCROcounter}, we can see $O(\sigma_E) \neq \mathcal{T}_E(O(\sigma_E))$. However, for a state, $\rho$, obtained from state preparation, $\rho$ must satisfy $\rho = \mathcal{T}_E(\rho)$. Hence, $O(\sigma_E)$ cannot be obtained from state preparation, which means $O$ cannot be replaced by a classical processing and state preparation. Thus, we prove the equivalence of the first and second criteria for $\mathrm{qcCRO}_{\{E_n\}}$.

Similar for $\mathrm{qqCRO}_{\{E_n\}}$, given any quantum operation $O$ satisfying $O = \mathcal{T}_E \circ O \circ \mathcal{T}_E$, we can replace it with measurement, classical processing and state preparation. The stochastic matrix of the corresponding classical processing is $T_{nm} = \tr(E_nO(\frac{E_m}{\tr(E_m)}))$. If an operation $O$ does not satisfy $O = \mathcal{T}_E \circ O \circ \mathcal{T}_E$, then $\exists \sigma$, $O(\sigma)\neq \mathcal{T}_E \circ O \circ \mathcal{T}_E(\sigma)$. If $O(\sigma) \neq \mathcal{T}_E \circ O(\sigma)$, then $O(\sigma)$ cannot be obtained with state preparation, which means $O$ cannot be replaced. If $O(\sigma) = \mathcal{T}_E \circ O(\sigma)$, then $\mathcal{T}_E \circ O(\sigma) \neq \mathcal{T}_E \circ O \circ \mathcal{T}_E(\sigma)$, implying $O$ is not a $\mathrm{qcCRO}_{\{E_n\}}$. As any $\mathrm{qqCRO}_{\{E_n\}}$ is a $\mathrm{qcCRO}_{\{E_n\}}$, in this case $O$ cannot be replaced neither. Here, we  prove the equivalence of the first and second criteria for $\mathrm{qqCRO}_{\{E_n\}}$ and complete the whole proof.

\section{Application of Classical Replacement in VQA}\label{appendsc:vqa}
Here, we present the details of applying classical replacement to the variational quantum algorithm. First we introduce the $n$-qubit Pauli group and Clifford group for further elaboration.

The four Pauli matrices on a qubit are defined as
\begin{equation}
\mathbb{I}_2 = \begin{pmatrix}
1 & 0 \\
0 & 1
\end{pmatrix},
\quad
X = \begin{pmatrix}
	0 & 1 \\
	1 & 0
\end{pmatrix},
\quad
Y = \begin{pmatrix}
	0 & -i \\
	i & 0
\end{pmatrix},
\quad
Z = \begin{pmatrix}
	1 & 0 \\
	0 & -1
\end{pmatrix}.
\end{equation}
The $n$-qubit Pauli group $\textsf{P}_n$ contains all the $n$-qubit Pauli operators in the form of tensor product of $n$ Pauli matrices,
\begin{equation}
\textsf{P}_n = \bigotimes_{i \in [n]} \{\mathbb{I}_2,X,Y,Z\}.
\end{equation}
The $n$-qubit Clifford group $\textsf{C}_n$ is defined as the normalizer group of $\textsf{P}_n$, that is, $\forall C\in \textsf{C}_n$, $\forall i \in [4^n]$, $\exists j\in [4^n]$, s.t., $P_j = C^{-1} P_i C$. In another word,
\begin{equation}
C^{-1}\textsf{P}_n C = \{C^{-1}P_iC \in \textsf{P}_n, P_i \in \textsf{P}_n\} = \{P_j \in \textsf{P}_n\} = \textsf{P}_n.
\end{equation}
As a consequence, given $\mathcal{T}_i = \frac{1}{2^{n-1}}(\tr(P_i^+ \cdot)P_i^+ + \tr(P_i^- \cdot)P_i^-)$ associated with $P_i$-measurement, it will become $\mathcal{T}_j$ associated with $P_j$-measurement under the action of $C$ as shown below.
\begin{equation}
\begin{split}
C^{-1} \circ \mathcal{T}_i \circ C (\rho) &= \frac{1}{2^{n-1}}(\tr(P_i^+ C\rho C^{-1}) C^{-1} P_i^+ C + \tr(P_i^- C\rho C^{-1}) C^{-1} P_i^- C)\\
&= \frac{1}{2^{n-1}}(\tr(C^{-1} P_i^+ C\rho) C^{-1} P_i^+ C + \tr(C^{-1} P_i^- C\rho) C^{-1} P_i^- C)\\
&= \frac{1}{2^{n-1}}(\tr(P_j^+ \rho) P_j^+ + \tr(P_j^- \rho) P_j^-)\\
&= \mathcal{T}_j (\rho),
\end{split}
\end{equation}
where $P_j = C^{-1} P_i C$. Here, we utilise $C^{-1} P_i^+ C = C^{-1} (\mathbb{I} + P_i) C / 2 = (\mathbb{I} + P_j) / 2 = P_j^+$ and $C^{-1} P_i^- C = C^{-1} (\mathbb{I} - P_i) C / 2 = (\mathbb{I} - P_j) / 2 = P_j^-$, where $\mathbb{I} = \mathbb{I}_2^{\otimes n}$.

In the variational quantum algorithm, to estimate the lowest energy level of a Hamiltonian, $H$, one first prepares a parametrized state $\rho = \ketbra{\psi}$ through a structured quantum circuit, and then evaluate $E = \tr(\rho H)$ by the measurement. Here, we take the circuit ansatz to be
\begin{equation}
\ket{\psi} = C(\vec{\alpha}) U(\vec{\theta}) \ket{0}^{\otimes n},
\end{equation}
where $n$ is the number of qubits, $\vec{\alpha}$ and $\vec{\theta}$ are parameter vectors. The energy for state $\rho$ is a function of $\vec{\alpha}$ and $\vec{\theta}$, that is, $E = E(\vec{\alpha}, \vec{\theta})$. By adjusting parameters $\vec{\alpha}$ and $\vec{\theta}$ one can achieve a minimum of $E(\vec{\alpha}, \vec{\theta})$ and set it as the approximate value for the lowest energy of $H$. In practice, $E(\vec{\alpha}, \vec{\theta})$ is often obtained with Pauli measurements. One may first express Hamiltonian $H$ as the sum of Pauli operators, $H = \sum_{i \in I} c_i P_i$ where $I\subseteq [4^n]$ is an index set and $P_i\in \textsf{P}_n$. Then, one can measure $\rho$ with $P_i$ to obtain $\tr(\rho P_i)$ for any $i$ in $I$ and evaluate $E = \tr(\rho H) = \sum_i c_i \tr(\rho P_i)$.

Here, the circuit ansatz is divided into two parts as shown in Figure~\ref{fig:CROVQA}. The part $U(\vec{\theta})$ is in general irreplaceable and needs real implementation while $C(\vec{\alpha})$ is taken from the set of qcCRO and can be moved after measurement and become classical processing. Recalling Lemma~\ref{lemma:qcCROext2}, a quantum operation before $P_i$-measurement that can be replaced by $P_j$-measurement and classical processing belongs to the set
\begin{equation}
\mathcal{R}_{ij} = \{O\in \mathrm{CPTP}| \mathcal{T}_i\circ O = \mathcal{T}_i\circ O\circ \mathcal{T}_j \}.
\end{equation}
Note that, for any $i\in I$, there exists $j\in [4^n]$,  $C(\vec{\alpha})$ can be moved after $P_i$-measurement and be replaced by $P_j$-measurement followed with classical processing. Then we deduce that $C(\vec{\alpha})$ belongs to the set
\begin{equation}
\begin{split}
\mathcal{R} &= \bigcup_{j \in [4^n]} \bigcap_{i \in I} \mathcal{R}_{ij}\\
&= \{ O\in \mathrm{CPTP} | \exists j\in [4^n], \forall i\in I, \mathcal{T}_i \circ O = \mathcal{T}_i \circ O \circ \mathcal{T}_j \}.
\end{split}
\end{equation}
In fact, $\mathcal{R}$ is the maximal set of quantum operations that can be replaced with classical processing in this case. That means, the optimization of $C(\vec{\alpha})$ is confined in $\mathcal{R}$. As a consequence, one can improve the expressibility of the variational quantum circuit with the help of virtual quantum gates $C(\vec{\alpha})$ while only $U(\vec{\theta})$ needs to be really implemented.

\begin{figure}[!htbp]
	\centering
	\resizebox{12cm}{!}{\includegraphics{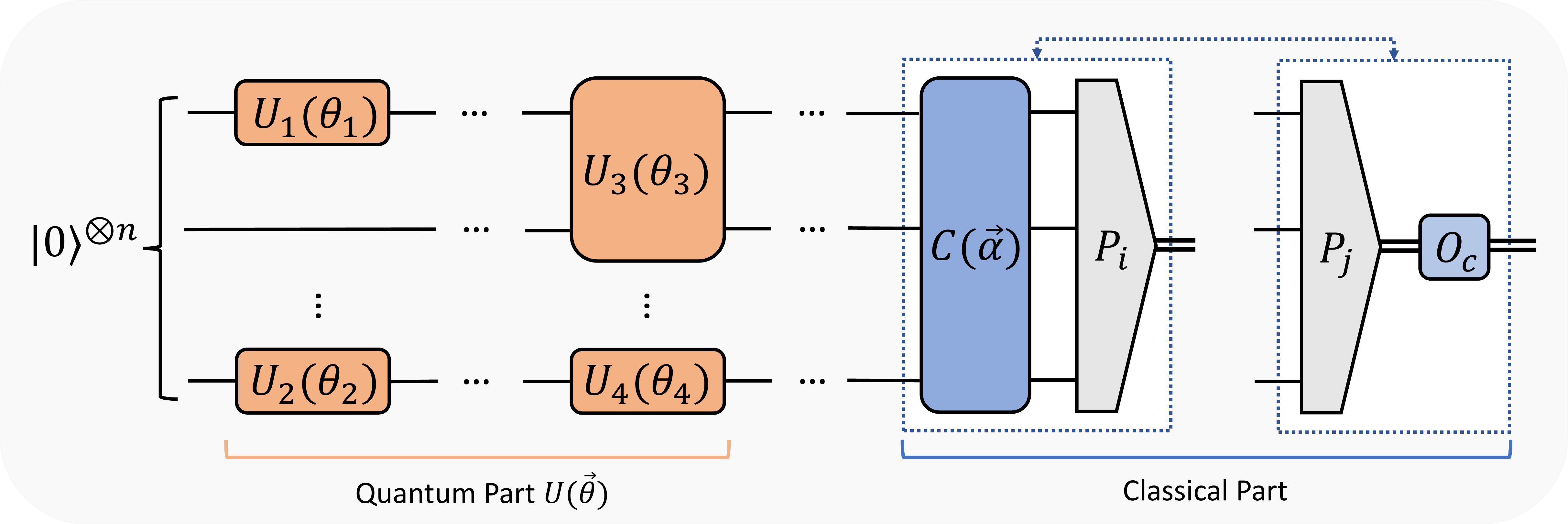}}
	\caption{The ansatz for a variational quantum circuit is set as $C(\vec{\alpha}) U(\vec{\theta}) \ket{0}^{\otimes n}$. It contains two parts. One is irreplaceable quantum operation $U(\vec{\theta})$ and the other is classically replaceable operation $C(\vec{\alpha})$. $\exists j\in [4^n]$, $\forall i\in I$, $C(\vec{\alpha})$ before $P_i$-measurement can be replaced with $P_j$-measurement and classical processing $O_c$.}
	\label{fig:CROVQA}
\end{figure}

In the main-text, we show that $\textsf{C}_n$ is a subset of $\mathcal{R}$. Moreover, replacing a Clifford gate does not acquire further classical processing after the measurement. It takes no price to implement a Clifford gate if we only concern the classical result of the Pauli measurement. In~\cite{ZhongXia2021SHVQA}, the authors propose this idea from another perspective, by viewing the evolution of $C(\vec{\alpha})$ in Heisenberg picture. The final measurement result can be given by
\begin{equation}
\begin{split}
E(\vec{\alpha}, \vec{\theta}) &= \bra{0}^{\otimes n} U^{-1}(\vec{\theta}) C^{-1}(\vec{\alpha}) H C(\vec{\alpha}) U(\vec{\theta}) \ket{0}^{\otimes n}\\
&= \sum_{i\in I} c_i \bra{0}^{\otimes n} U^{-1}(\vec{\theta}) C^{-1}(\vec{\alpha}) P_i C(\vec{\alpha}) U(\vec{\theta}) \ket{0}^{\otimes n}\\
&= \sum_{i\in I} c_i \bra{0}^{\otimes n} U^{-1}(\vec{\theta}) P_j U(\vec{\theta}) \ket{0}^{\otimes n}.
\end{split}
\end{equation}
Therefore, one can perform $P_j$-measurement on $U(\vec{\theta}) \ket{0}^{\otimes n}$ instead of performing $P_i$-measurement on $C(\vec{\alpha}) U(\vec{\theta}) \ket{0}^{\otimes n}$ to evaluating $E(\vec{\alpha}, \vec{\theta})$. The action of $C(\vec{\alpha})$ is moved from the state to Hamiltonian $H$ and hence $C(\vec{\alpha})$ becomes a virtual gate.

%\section{Proof of Theorems and Lemmas in Section~\ref{sc:Resource}}\label{thmproof:resource}
\section{Robustness of Irreplaceability}\label{thmproof:resource}
\subsection{Proof of Lemma~\ref{lemma:equiv}}\label{ProofLemmaEquiv}
Ahead of the proof, we briefly restate Lemma~\ref{lemma:equiv} for better reading. There are five equivalent definitions for the robustness of irreplaceability of a quantum channel $\mathcal{N}$ as shown below.
\begin{align}
\label{eq:RobofIrDef1'} R(\mathcal{N})
&= \min_{\mathcal{M}\in \mathrm{CPTP} } \left\{s\geq 0 \bigg| \frac{\Phi_{\mathcal{N}}+s\Phi_{\mathcal{M}}}{1+s}\in \mathbf{F} \right\}\\
\label{eq:RobofIrDef2'} &= \min_{\mathcal{M}\in \mathrm{CPTP}} \left\{s\geq 0 \bigg| \frac{\Phi_{\Delta\circ\mathcal{N}}+s\Phi_{\Delta\circ\mathcal{M}}}{1+s}\in \mathbf{F}'\right\}\\
\label{eq:RobofIrDef3'} &= \min_{\mathcal{M}\in \mathrm{CPTP} } \left\{s\geq 0 \bigg| \frac{\mathcal{N}+s\mathcal{M}}{1+s}\in \mathrm{qcCRO}\right\}\\
\label{eq:RobofIrDef4'} &= \min_{\mathcal{M}\in \mathrm{CPTP} } \left\{s\geq 0 \bigg| \frac{\Delta\circ\mathcal{N}+s\Delta\circ\mathcal{M}}{1+s}\in \mathrm{qcCRO}\right\}\\
\label{eq:RobofIrDef5'} &= \min_{\mathcal{M}\in \mathrm{CPTP} } \left\{s\geq 0 \bigg| \frac{\Delta\circ\mathcal{N}+s\mathcal{M}}{1+s}\in \mathrm{qcCRO}\right\},
\end{align}
where $\mathbf{F} = \{\Phi_{\mathcal{M}}\big| \mathcal{M}\in \mathrm{qcCRO}\}$ and $\mathbf{F}'=\{\Phi_{\Delta\circ\mathcal{M}}\big| \mathcal{M}\in \mathrm{qcCRO}\}$.

\begin{proof}
The equivalence between Eq.~\eqref{eq:RobofIrDef1'} and Eq.~\eqref{eq:RobofIrDef2'} can be directly verified by the definition. The equivalence of Eq.~\eqref{eq:RobofIrDef1'} and Eq.~\eqref{eq:RobofIrDef3'} comes from the one-to-one correspondence of a channel and its Choi-state. We further prove the equivalence between Eq.~\eqref{eq:RobofIrDef3'} and Eq.~\eqref{eq:RobofIrDef4'}, Eq.~\eqref{eq:RobofIrDef3'} and Eq.~\eqref{eq:RobofIrDef5'} respectively to prove Lemma~\ref{lemma:equiv}, which comes from the idempotence of the dephasing operation.

Notice that $\Delta\circ \Delta = \Delta$, we have
\begin{equation}\label{eq:Robustequ}
\Delta \circ \frac{\mathcal{N}+s\mathcal{M}}{1+s} = \Delta \circ \frac{\Delta \circ \mathcal{N}+s\Delta \circ \mathcal{M}}{1+s}.
\end{equation}
Assuming the minimization of Eq.~\eqref{eq:RobofIrDef3'} gives $s_1$ and that of Eq.~\eqref{eq:RobofIrDef4'} gives $s_2$, that means
\begin{equation}\label{eq:Robustequ2}
\Delta \circ \frac{\mathcal{N}+s_1\mathcal{M}}{1+s_1} = \Delta \circ \frac{\mathcal{N}+s_1\mathcal{M}}{1+s_1} \circ \Delta.
\end{equation}
Combining Eq.~\eqref{eq:Robustequ} with Eq.~\eqref{eq:Robustequ2} we get
\begin{equation}\label{eq:Robustequ3}
\Delta \circ \frac{\Delta \circ \mathcal{N}+s_1\Delta \circ \mathcal{M}}{1+s_1} = \Delta \circ \frac{\Delta \circ\mathcal{N}+s_1\Delta \circ\mathcal{M}}{1+s_1} \circ \Delta.
\end{equation}
From the minimization of robustness we know $s_2\leq s_1$. Reversely, combining Eq.~\eqref{eq:Robustequ} with Eq.~\eqref{eq:Robustequ3} while substituting $s_1$ with $s_2$, we can prove $s_1\leq s_2$. Then $s_1 = s_2$ implying the equivalence between Eq.~\eqref{eq:RobofIrDef3'} and Eq.~\eqref{eq:RobofIrDef4'}. Similarly, we assume the minimization of Eq.~\eqref{eq:RobofIrDef5'} gives $s_3$ and from $\Delta\circ \Delta = \Delta$ we have
\begin{equation}\label{eq:Robustequ4}
	\Delta \circ \frac{\mathcal{N}+s\mathcal{M}}{1+s} = \Delta \circ \frac{\Delta \circ \mathcal{N}+s\mathcal{M}}{1+s}.
\end{equation}
Combining Eq.~\eqref{eq:Robustequ2} with Eq.~\eqref{eq:Robustequ4} we get
\begin{equation}
	\Delta \circ \frac{\Delta \circ \mathcal{N}+s_1 \mathcal{M}}{1+s_1} = \Delta \circ \frac{\Delta \circ\mathcal{N}+s_1\mathcal{M}}{1+s_1} \circ \Delta.
\end{equation}
Then we get $s_3\leq s_1$. Following the similar approach we can also prove $s_1\leq s_3$ implying $s_1 = s_3$. Thus, we prove the equivalence between Eq.~\eqref{eq:RobofIrDef3'} and Eq.~\eqref{eq:RobofIrDef5'}, and complete the whole proof.
\end{proof}

\subsection{Proof of Lemma~\ref{lemma:robust}}\label{ProofLemmaRobust}
\begin{proof}
First, we prove the monotonicity of robustness of irreplaceability. For any free superchannel $\Lambda\in \mathcal{F}$, suppose $\mathcal{M}$ is the channel achieving the minimal value of $s$, s.t., $\frac{\mathcal{N}+s\mathcal{M}}{1+s}\in\mathrm{qcCRO}$. Then
\begin{equation}
	\Lambda\left(\frac{\mathcal{N}+s\mathcal{M}}{1+s}\right) = \frac{\Lambda(\mathcal{N})+s\Lambda(\mathcal{M})}{1+s}\in\mathrm{qcCRO}.
\end{equation}
As $\Lambda(\mathcal{M})$ is also a CPTP channel, from the definition we see that $R(\Lambda(\mathcal{N})) \leq s = R(\mathcal{N})$.

The convexity of robustness comes from the convexity of qcCRO. Given a set of channels $\{\mathcal{N}_i\}$ where $i$ belongs to an index set $I$, assuming $s_i = R(\mathcal{N}_i)$ and $\mathcal{M}_i$ achieves the minimal value, that is,
\begin{equation}
	\mathcal{M}_i^{\prime} = \frac{\mathcal{N}_i+s_i\mathcal{M}_i}{1+s_i} \in \mathrm{qcCRO}.
\end{equation}
We set $\mathcal{M} = \frac{\sum_i p_is_i\mathcal{M}^{\prime}_i}{s}$, $s = \sum_{i}p_is_i$. Due to the convexity of qcCRO,
\begin{equation}
	\frac{\sum_{i}p_i\mathcal{N}_i+s\mathcal{M}}{1+s} = \frac{\sum_ip_i(1+s_i)\mathcal{M}^{\prime}_i}{1+s} \in \mathrm{qcCRO}.
\end{equation}
From the definition of the robustness we get $R(\sum_{i}p_i\mathcal{N}_i)\leq s = \sum_{i}p_i R(\mathcal{N}_i)$.
\end{proof}

\subsection{Proof of Lemma~\ref{lemma:extension}}\label{ProofLemmaExtension}
Before proving Lemma~\ref{lemma:extension}, we introduce some basic properties of CROs in a large system with multiple parts. These properties are natural corollaries of Theorem~\ref{thm:CROequiv}. We first consider the joint behavior of independent CROs on different subsystems, which is described by the tensor-product operation.

\begin{corollary}\label{coro:tensor}
cqCRO, qqCRO, and qcCRO are closed under the action of tensor product.
\end{corollary}

\begin{proof}
We prove the result for qcCRO. Similar arguments apply to cqCRO and qqCRO. Recall that the computational basis of an $n$-party system, $\bigotimes_{k=0}^{n-1} \mathcal{H}_k$, is the tensor product of subsystem computational bases. Thus, the dephasing operation on $\bigotimes_{k=0}^{n-1} \mathcal{H}_k$ is also the tensor product of subsystem dephasing operations,
\begin{equation}
\Delta_{[n]} = \bigotimes_{k=0}^{n-1} \Delta_k.
\end{equation}
For any $k\in [n]$, given qcCRO $O_k$ on subsystem $\mathcal{H}_k$, $\bigotimes_{k=0}^{n-1} O_k$ satisfies
\begin{equation}
\Delta_{[n]} \circ \bigotimes_{k=0}^{n-1} O_k = \bigotimes_{k=0}^{n-1} \Delta_k\circ O_k = \bigotimes_{k=0}^{n-1} \Delta_k\circ O_k\circ \Delta_k = \Delta_{[n]} \circ \bigotimes_{k=0}^{n-1} O_k\circ \Delta_{[n]}.
\end{equation}
It means that $\bigotimes_{k=0}^{n-1} O_k$ is a qcCRO on $\bigotimes_{k=0}^{n-1} \mathcal{H}_k$.
\end{proof}

Next, we consider the inverse operation of the tensor product. Without loss of generality, we consider a bipartite system $\mathcal{H}_0\otimes \mathcal{H}_1$. The dephasing operations on $\mathcal{H}_0$, $\mathcal{H}_1$, and $\mathcal{H}_0\otimes \mathcal{H}_1$ are denoted as $\Delta_{0}$, $\Delta_{1}$, and $\Delta_{01} = \Delta_{0}\otimes \Delta_1$, respectively. Consider a quantum operation $O_{01}$ on the bipartite system. If its effective operation on a subsystem is CPTP, then we call the operation on the subsystem as the partial trace over the other subsystem. The partial trace over $\mathcal{H}_1$ or $\mathcal{H}_0$ of $O_{01}$ is denoted as $O_0$ or $O_1$, respectively. In general, $O_{01}$ jointly acts on the two subsystems and we cannot write it as $O_{01}=O_0\otimes O_1$. While in our cases, it suffices to consider the following definitions for the partial trace of quantum operations,
\begin{equation}\label{eq:PartialOperation}
\begin{split}
    \forall \rho_0 \in \mathcal{D}(\mathcal{H}_0),\quad O_0(\rho_0) &= \tr_1\left(O_{01}\left(\rho_0\otimes \frac{\mathbb{I}_1}{d_1}\right)\right),\\
    \forall \rho_1 \in \mathcal{D}(\mathcal{H}_1),\quad O_1(\rho_1) &= \tr_0\left(O_{01}\left(\frac{\mathbb{I}_0}{d_0}\otimes \rho_1\right)\right),
\end{split}
\end{equation}
where $\mathbb{I}_0$ and $\mathbb{I}_1$ are identity operators on $\mathcal{H}_0$ and $\mathcal{H}_1$, respectively, $d_0 = \dim \mathcal{H}_0$, and $d_1 = \dim \mathcal{H}_1$. The choice of a maximally mixed state of subsystem $\mathcal{H}_1$ in the definition of $O_0$ closely relates to our definition of the irreplaceability measures, where we essentially use the Choi states. The definition in Eq.~\eqref{eq:PartialOperation} ensures that the partial trace over system $\mathcal{H}_1$ of Choi state $\Phi_{O_{01}}$ is equal to Choi state $\Phi_{O_{0}}$. Operationally, this choice can be interpreted as the average effect on subsystem $\mathcal{H}_0$ when the subsystem $\mathcal{H}_1$ is prepared in a uniformly random pure state according to the Haar measure. This is also the case for the definition of $O_1$.

With respect to the partial-trace operations in Eq.~\eqref{eq:PartialOperation}, we have the following corollary from Theorem~\ref{thm:CROequiv}.

\begin{corollary}\label{coro:partrace}
cqCRO, qqCRO, and qcCRO are closed under the action of partial trace.
\end{corollary}

\begin{proof}
We prove the result for qcCRO. Similar arguments apply to cqCRO and qqCRO. If $O_{01}$ is a qcCRO, which means $\Delta_{01}\circ O_{01} = \Delta_{01}\circ O_{01}\circ \Delta_{01}$, then $\forall \rho_0\in \mathcal{D}(\mathcal{H}_0)$,
\begin{equation}
\begin{split}
\Delta_{0}\circ O_0(\rho_0) &= \Delta_0 \left(\tr_1\left(O_{01}\left(\rho_0\otimes \frac{\mathbb{I}_1}{d_1}\right)\right)\right)\\
&= \Delta_0 \left(\tr_1\left( (I_0\otimes\Delta_1)\circ O_{01}\left(\rho_0\otimes \frac{\mathbb{I}_1}{d_1}\right)\right)\right)\\
&= \tr_1\left( \Delta_{01}\circ O_{01}\left(\rho_0\otimes \frac{\mathbb{I}_1}{d_1}\right)\right)\\
&= \tr_1\left( \Delta_{01}\circ O_{01}\circ\Delta_{01}\left(\rho_0\otimes \frac{\mathbb{I}_1}{d_1}\right)\right)\\
&= \Delta_0\left(\tr_1\left( (I_0\otimes\Delta_1)\circ O_{01}\left(\Delta_0(\rho_0)\otimes \frac{\mathbb{I}_1}{d_1}\right)\right)\right)\\
&= \Delta_0\left(\tr_1\left( O_{01}\left(\Delta_0(\rho_0)\otimes \frac{\mathbb{I}_1}{d_1}\right)\right)\right)\\
%&= \Delta_0(O_0\Delta_0(\rho_0))\\
&= \Delta_0\circ O_0\circ \Delta_0 (\rho_0).
\end{split}
\end{equation}
This directly leads to $\Delta_{0}\circ O_0 = \Delta_{0}\circ O_0\circ \Delta_0$. Utilising the symmetry, we also obtain $\Delta_{1}\circ O_1 = \Delta_{1}\circ O_1\circ \Delta_1$. In the following, we write $I_0\otimes\Delta_1$ as $\Delta_1$ for simplicity and similarly for the alike.
\end{proof}

With the above two results, we now prove Lemma~\ref{lemma:extension}.

\begin{proof}
Given a quantum channel on system $\mathcal{H}_0$, $\mathcal{N}$, and the identity operation on ancillary system $\mathcal{H}_1$, $I$, we denote $s_1 = R(\mathcal{N})$ and $s_2 = R(\mathcal{N}\otimes I)$. Then by definition, there exists a quantum channel on $\mathcal{H}_0$, $\mathcal{M}_0$, and a quantum channel on $\mathcal{H}_0\otimes \mathcal{H}_1$, $\tilde{\mathcal{M}}_{01}$, satisfying
\begin{align}
\label{eq:subsysrobust} \Delta_0 \circ \frac{\mathcal{N}+s_1\mathcal{M}_0}{1+s_1} &= \Delta_0 \circ \frac{\mathcal{N}+s_1\mathcal{M}_0}{1+s_1}\circ \Delta_0,\\
\label{eq:wholesysrobust} \Delta_{01} \circ \frac{\mathcal{N}\otimes I+s_2\tilde{\mathcal{M}}_{01}}{1+s_2} &= \Delta_{01} \circ \frac{\mathcal{N}\otimes I+s_2\tilde{\mathcal{M}}_{01}}{1+s_2}\circ \Delta_{01}.
\end{align}
From Eq.~\eqref{eq:subsysrobust} one can deduce that
\begin{equation}
\begin{split}
\Delta_{01} \circ \frac{\mathcal{N}\otimes I+s_1\mathcal{M}_{0}\otimes I}{1+s_1} &=  \left(\Delta_0 \circ \frac{\mathcal{N}+s_1\mathcal{M}_0}{1+s_1}\right) \otimes (\Delta_1\circ I)\\
&= \left(\Delta_0 \circ \frac{\mathcal{N}+s_1\mathcal{M}_0}{1+s_1}\circ \Delta_0\right) \otimes (\Delta_1\circ I\circ \Delta_1)\\
&= \Delta_{01} \circ \frac{\mathcal{N}\otimes I+s_1\mathcal{M}_{0}\otimes I}{1+s_1}\circ \Delta_{01},
\end{split}
\end{equation}
which implies that $s_2\leq s_1$. On the other hand, one can take partial trace over system $\mathcal{H}_1$ for both sides of Eq.~\eqref{eq:wholesysrobust} and obtain that
\begin{equation}
\Delta_{0} \circ \frac{\mathcal{N}+s_2\tilde{\mathcal{M}}_{0}}{1+s_2} = \Delta_{0} \circ \frac{\mathcal{N}+s_2\tilde{\mathcal{M}}_{0}}{1+s_2}\circ \Delta_{0},
\end{equation}
which implies that $s_1\leq s_2$. Here, $\tilde{\mathcal{M}}_0$ represents the partial trace of $\tilde{\mathcal{M}}_{01}$ over $\mathcal{H}_1$ given by Eq.~\eqref{eq:PartialOperation}. Therefore, we obtain $R(\mathcal{N}) = s_1 = s_2 = R(\mathcal{N}\otimes I)$ and complete the proof.
\end{proof}

\subsection{Proof of Theorem~\ref{thm:nonlocal}}\label{Append:ProofThmNonlocal}
\begin{proof}
The proof idea is transforming the minimization problem in robustness of irreplaceability to a dual problem. The dual problem is a maximization problem, where the objective function corresponds to the performance of a channel in the non-local game. For further elaboration, we clarify the notations. The Choi-state $I\otimes \mathcal{N}(\ketbra{\Phi^+})$ of a channel $\mathcal{N}$ is a bipartite state and we label the two parties with 0 and 1, respectively. That is, $I\otimes \mathcal{N}(\ketbra{\Phi^+})\in\mathcal{H}_0\otimes \mathcal{H}_1$. The computational basis of each party $\mathcal{H}_k$ is denoted as $\{\ket{i_k}, i\in[d]\}$, $k = 0,1$. With respect to the computation basis, we define two dephasing operations on the subsystem, such that $\forall \rho\in \mathcal{D}(\mathcal{H}_0 \otimes \mathcal{H}_1)$,
\begin{align}
\label{eq:delta0}	&\Delta_{0} = \sum_{i=0}^{d-1} \ketbra{i_0} \otimes \tr_0(\ketbra{i_0} \rho),\\
\label{eq:delta1}	&\Delta_{1} = \sum_{j=0}^{d-1} \tr_1(\ketbra{j_1} \rho) \otimes \ketbra{j_1}.
\end{align}
For the joint system $\mathcal{H}_0\otimes \mathcal{H}_1$, we take the computational basis as $\{\ket{i_0j_1}, i_0,j_1\in[d]\}$ and represent the dephasing operation on $\mathcal{H}_0\otimes \mathcal{H}_1$ as
\begin{equation}\label{eq:delta01}
\Delta_{01} = \sum_{i_0,j_1=0}^{d-1}\tr(\ketbra{i_0j_1}\rho)\ketbra{i_0j_1}.
\end{equation}

Now we take the second equivalent definition of robustness of irreplaceability, i.e., Eq.~\eqref{eq:RobofIrDef2}, and express the minimization problem as a conic optimization problem~\cite{Takagi2019Robustness},
\begin{equation}
\begin{split}
&R(\mathcal{N}) = \min_{s, A} \ s, \\
s.t. \ &A - \Delta_{1}(\Phi_{\mathcal{N}}) \in \mathrm{cone}(\mathbf{S}), \\
&A\in \mathrm{cone}(\mathbf{F}'), \\
&\tr_1(A) = (1+s)\frac{\mathbb{I}_0}{d},
\end{split}
\end{equation}
where $\mathbf{S} = \{\Delta_{1}(\Phi_{\mathcal{N}})$, $\mathcal{N}\in \mathrm{CPTP}\}$, $\mathbf{F}' = \{\Delta_{1}(\Phi_{\mathcal{M}}), \mathcal{M}\in \mathrm{qcCRO}\}$, and $\mathrm{cone}(\cdot)$ represents their unnormalized versions, which forms a convex cone. The term $A$ corresponds to the numerator, $\Phi_{\Delta\circ\mathcal{N}}+s\Phi_{\Delta\circ\mathcal{M}}$, in Eq.~\eqref{eq:RobofIrDef2}. Note that $\Phi_{\Delta\circ\mathcal{N}} = \Delta_{1}(\Phi_{\mathcal{N}})$. We use $\mathbb{I}_0\in\mathcal{L}(\mathcal{H}_0)$ to denote the identity operator, where $\mathcal{L}(\mathcal{H}_0)$ represents the set of linear operators on $\mathcal{H}_0$. The Lagrangian function of this problem is
\begin{equation}
L(s, A, W, X, Y) = s - \langle W, A-\Delta_{1}(\Phi_{\mathcal{N}})\rangle  - \langle X, A\rangle  - \left\langle Y, (1+s)\frac{\mathbb{I}_0}{d}-\tr_1(A)\right\rangle.
\end{equation}
Here, $\langle \cdot,\cdot\rangle $ represents Hilbert-Schmidt inner product, with
\begin{equation}
  \langle A,B\rangle=\tr(A^\dag B),\forall A,B\in\mathcal{L}(\mathcal{H}),
\end{equation}
and $W\in (\mathrm{cone}(\mathbf{S}))^*$, $X\in (\mathrm{cone}(\mathbf{F}'))^*$, $Y\in\mathcal{L}(\mathcal{H}_0)$. Here, $(\mathrm{cone}(\mathbf{S}))^*$ represents the dual cone of $\mathrm{cone}(\mathbf{S})$, defined as $(\mathrm{cone}(\mathbf{S}))^*=\{V\in\mathcal{L}(\mathcal{H}_0)|\forall U\in\mathrm{cone}(\mathbf{S}),\langle U,V\rangle\geq0\}$, and the definition of $(\mathrm{cone}(\mathbf{F}'))^*$ is similar. The dual problem is
\begin{equation}
\begin{split}
&\max_{W, X, Y} \min_{s, A} L(s, A, W, X, Y), \\
s.t. \ &\langle W, B\rangle  \geq 0, \forall B\in \mathbf{S},\\
&\langle X, C\rangle  \geq 0, \forall C\in \mathbf{F}'.
\end{split}
\end{equation}
We can change the form of Lagrangian function:
\begin{equation}
L(s, A, W, X, Y) = s\left(1-\left\langle Y, \frac{\mathbb{I}_0}{d}\right\rangle \right) - \langle W + X - Y\otimes \mathbb{I}_1, A\rangle  + \langle W, \Delta_{1}(\Phi_{\mathcal{N}})\rangle  - \left\langle Y, \frac{\mathbb{I}_0}{d}\right\rangle.
\end{equation}
Note that $\langle Y, \tr_1(A)\rangle  = \langle Y\otimes \mathbb{I}_1, A\rangle $. Then the dual problem is equivalent to:
\begin{equation}
\begin{split}
&\max_{W, X, Y} \ \langle W, \Delta_{1}(\Phi_{\mathcal{N}})\rangle  - \left\langle Y, \frac{\mathbb{I}_0}{d}\right\rangle, \\
s.t. \ &1 - \left\langle Y, \frac{\mathbb{I}_0}{d}\right\rangle  = 0, \\
&W + X - Y\otimes \mathbb{I}_1 = 0,\\
&\langle W, B\rangle  \geq 0, \forall B\in \mathbf{S},\\
&\langle X, C\rangle  \geq 0, \forall C\in \mathbf{F}'.
\end{split}
\end{equation}
The above can also be written as
\begin{equation}
\begin{split}
&\max_{W, X, Y} \ \tr(W\Delta_{1}(\Phi_{\mathcal{N}})) - 1,\\
s.t. \ &\tr(Y) = d,\\
&W = Y\otimes \mathbb{I}_1 - X,\\
&\tr(WB) \geq 0, \forall B\in \mathbf{S},\\
&\tr(WC) \leq \tr(Y\otimes \mathbb{I}_1 C) = \tr(\frac{Y}{d}) = 1, \forall C\in \mathbf{F}'.
\end{split}
\end{equation}
By taking $X = \frac{\mathbb{I}_{01}}{2}$, $Y = \mathbb{I}_0$, $W = \frac{\mathbb{I}_{01}}{2}$, we can check that Slater's conditions are fulfilled and hence the problem satisfies the strong duality~\cite{boyd2004convex}. Note that the primal problem is minimizing a convex function $s$ and $\langle Y, (1+s)\frac{\mathbb{I}_0}{d}-\tr_1(A)\rangle $ is an affine function relative to $s$ and $A$. Also, we can see that the objective function is not related to $X$ and $Y$, hence we can rewrite the optimization as
\begin{equation}\label{robustness_opt}
	\begin{split}
		&R(\mathcal{N}) = \max_{W\in (\mathrm{cone}(\mathbf{S}))^*} \ \tr(W\Delta_{1}(\Phi_{\mathcal{N}})) - 1,\\
		s.t. \ &\tr(WB) \geq 0, \forall B\in \mathbf{S},\\
		&\tr(WC) \leq 1, \forall C\in \mathbf{F}'.
	\end{split}
\end{equation}
Note that for any qcCRO $\mathcal{M}$, $\Delta_{1}(\Phi_{\mathcal{M}})\in\mathbf{F}'$, there always exists $W = \Delta_{1}(\Phi_{\mathcal{M}})$ satisfying the restriction and achieving the maximization. The corresponding robustness is 0 which is the value for any qcCRO.

On the other side, the performance of channel $\mathcal{N}$ in the non-local game is
\begin{equation}\label{performace-robustness}
\begin{split}
p(\mathcal{N}, \{\alpha_{ij}\}, \{\sigma_i\}, \{\ketbra{j}\}) &= \sum_{i,j} \alpha_{ij} \tr(\mathcal{N}(\sigma_i) \ketbra{j})\\
&= d\sum_{i,j} \alpha_{ij} \tr(\Phi_{\mathcal{N}} \sigma_i^T \otimes \ketbra{j})\\
&= d\sum_{i,j} \alpha_{ij} \tr(\Delta_{1}(\Phi_{\mathcal{N}}) \sigma_i^T \otimes \ketbra{j})\\
&= \tr(W\Delta_{1}(\Phi_{\mathcal{N}})),
\end{split}
\end{equation}
where $W = d\sum_{i,j} \alpha_{ij} \sigma_i^T\otimes \ketbra{j}\in (\mathrm{cone}(\mathbf{S}))^*$. Therefore, by varying the strategy $\alpha_{ij}$ and $\sigma_i$ we can obtain any $W$ satisfying $\tr(WB) \geq 0, \forall B\in \mathbf{S}$ and $\tr(WC) \leq 1, \forall C\in \mathbf{F}'$. Comparing Eq.~\eqref{eq:advantage2robust}, Eq.~\eqref{robustness_opt} and Eq.~\eqref{performace-robustness}, we relate the best advantage a quantum channel can provide to the robustness measure and complete the proof,
\begin{equation}
\max_{\{\alpha_{ij}\}, \{\sigma_i\}} \frac{\tr(W\Delta_{1}(\Phi_{\mathcal{N}}))}{\max_{\mathcal{M}\in\mathrm{qcCRO}}\ \tr(W\Delta_{1}(\Phi_{\mathcal{M}}))} = 1 + R(\mathcal{N}).
\end{equation}
\end{proof}

\subsection{Calculation of the Robustness of Irreplaceability}\label{Append:RobustCal}
Given a channel $\mathcal{N}\in \mathrm{CPTP}$, evaluating its robustness of irreplaceability can be recast into a convex optimisation problem. Here, we consider the case of a qcCRO. Suppose the robustness value is given by $R(\mathcal{N})=s^*$ and the corresponding CPTP map in its evaluation is given by $\mathcal{M}^*$. That is,
\begin{equation}
	\frac{\Phi_{\mathcal{N}}+s^*\Phi_{\mathcal{M}^*}}{1+s^*}\in \mathbf{F},
\end{equation}
where $\Phi_{\mathcal{N}}$ denotes the Choi state of $\mathcal{N}$, $\Phi_{\mathcal{N}} = I\otimes \mathcal{N} (\ketbra{\Phi^+})\in\mathcal{D}(\mathcal{H}_0\otimes \mathcal{H}_1)$, and $\mathbf{F} = \{\Phi_{\mathcal{M}}\big| \mathcal{M}\in \mathrm{qcCRO}\}$.
%As in Eq.~\eqref{eq:delta0}, \eqref{eq:delta1} and \eqref{eq:delta01}, we denote the first subsystem of the Choi state as $0$ and the second subsystem as $1$.
Then, we can rewrite the definition of robustness of irreplaceability as the following optimization problem,
\begin{equation}\label{eq:RobustCal}
	\begin{split}
		1+R(\mathcal{N}) &= \inf_{\psi} \tr(\psi), \\
		s.t. \  \psi-\Phi_{\mathcal{N}} &\geq 0, \\
		\psi&\geq0, \\
		\tr_1(\psi)&=\frac{\mathbb{I}_0}{d}\tr(\psi),\\
		\psi&\in\mathrm{cone}(\mathbf{F}),
	\end{split}
\end{equation}
where $d$ is the dimension of subsystem $0$ and $\mathbb{I}_0$ is the identity operator on subsystem $0$. The second and the third constraints correspond to the requirement of a CPTP map. For the set of qcCRO, the cone of its corresponding set of Choi states can be described as
\begin{equation}\label{eq:qcCROcone}
	\mathrm{cone}(\mathbf{F})=\{\psi|\Delta_1(\psi)=\Delta_{01}(\psi),\psi\geq0\},
\end{equation}
where $\Delta_1$ and $\Delta_{01}$ are given in Eq.~\eqref{eq:delta1} and \eqref{eq:delta01}. Eq.~\eqref{eq:qcCROcone} is the Choi-state representation of Eq.~\eqref{eq:qcCRO}. Then, Eq.~\eqref{eq:RobustCal} is written in a standard semi-definite programming problem, which can be solved in polynomial time, for example, by using interior-point algorithm.

\section{Relative Entropy of Irreplaceability}\label{append:rel}
The relative entropy of irreplaceability is defined as follows.

\begin{definition}[Relative Entropy of Irreplaceability] Given a channel $\mathcal{N}\in \mathrm{CPTP}$, the relative entropy of irreplaceability of $\mathcal{N}$ is
\begin{equation}\label{eq:RelofIrDef}
C_{\mathrm{rel}}(\mathcal{N}) = \min_{\mathcal{M}\in \mathrm{qcCRO}} D( \Phi_{\Delta\circ \mathcal{N}} \| \Phi_{\Delta\circ \mathcal{M}}),
\end{equation}
where $D(\rho\|\sigma) = \tr(\rho\log \rho)-\tr(\rho\log \sigma)$ is the relative entropy from $\rho$ to $\sigma$.
\end{definition}
The set of qcCRO and $\mathbf{F}'=\{\Phi_{\Delta\circ\mathcal{M}}\big| \mathcal{M}\in \mathrm{qcCRO}\}$ are convex, hence Eq.~\eqref{eq:RelofIrDef} is well-defined. This measure is different from the normal definition of the relative entropy measure, usually defined as the minimum relative entropy of a channel to the free channels, i.e., $\min_{\mathcal{M}\in \mathrm{qcCRO}} D( \Phi_{\mathcal{N}} \| \Phi_{\mathcal{M}})$. Here, we adopt the Choi state of $\Delta\circ \mathcal{N}$ instead of $\mathcal{N}$ as we concern about whether an operation followed with a measurement can be replaced by classical operations or not. The irreplaceability of a channel $\mathcal{N}$ should be the same as $\Delta\circ \mathcal{N}$ as no measurement can distinguish them. Also, we have proven that $R(\mathcal{N}) = R(\Delta\circ \mathcal{N})$, showing that it is reasonable to consider the distance between $\Delta\circ\mathcal{N}$ and $\Delta\circ\mathcal{M}$ instead of $\mathcal{N}$ and $\mathcal{M}$.

The relative entropy of irreplaceability has an analytic expression, as given by the following lemma. The dephasing operations $\Delta_{0},\Delta_{1},\Delta_{01}$ are given by Eq.~\eqref{eq:delta0}, \eqref{eq:delta1} and \eqref{eq:delta01}, respectively.
\begin{lemma}\label{lemma:rel}
The relative entropy of irreplaceability of a channel $\mathcal{N}$ equals to
\begin{equation}\label{eq:RelofIr}
	C_{\mathrm{rel}}(\mathcal{N}) = S(\Delta_{01}(\Phi_{\mathcal{N}})) - S(\Delta_{1}(\Phi_{\mathcal{N}})),
\end{equation}
where $S(\rho) = -\tr(\rho \log \rho)$ is the von-Neumann entropy.
\end{lemma}

\begin{proof}
For any qcCRO, $\mathcal{M}$,
\begin{equation}
\begin{split}
\Phi_{\Delta \circ \mathcal{M}} = \Delta_{1}(\Phi_{\mathcal{M}}) &= I\otimes (\Delta\circ \mathcal{M}) (\ketbra{\Phi^+})\\
&= \frac{1}{d} \sum_{i,j} I \otimes (\Delta\circ \mathcal{M} \circ \Delta) \ketbra{ii}{jj}\\
&= \frac{1}{d} \sum_i \ketbra{i} \otimes (\Delta\circ \mathcal{M})(\ketbra{i})\\
&= \Delta_{0}\left(\frac{1}{d} \sum_i \ketbra{i} \otimes (\Delta\circ\mathcal{M})(\ketbra{i})\right)\\
&= \Delta_{0} (I\otimes (\Delta\circ \mathcal{M}) (\ketbra{\Phi^+}))\\
&= \Delta_{01}(\Phi_{\mathcal{M}}).
\end{split}
\end{equation}
Then,
\begin{equation}
\begin{split}
D(\Delta_{1}(\Phi_{\mathcal{N}})||\Delta_{1}(\Phi_{\mathcal{M}})) &= -S(\Delta_{1}(\Phi_{\mathcal{N}})) - \tr(\Delta_{1}(\Phi_{\mathcal{N}})\log \Delta_{1}(\Phi_{\mathcal{M}}))\\
&= -S(\Delta_{1}(\Phi_{\mathcal{N}})) - \tr(\Delta_{1}(\Phi_{\mathcal{N}})\log \Delta_{01}(\Phi_{\mathcal{M}}))\\
&= -S(\Delta_{1}(\Phi_{\mathcal{N}})) - \tr(\Delta_{1}(\Phi_{\mathcal{N}}) \Delta_{01}(\log \Phi_{\mathcal{M}}))\\
&= -S(\Delta_{1}(\Phi_{\mathcal{N}})) - \tr(\Delta_{01}(\Phi_{\mathcal{N}}) \Delta_{01}(\log \Phi_{\mathcal{M}}))\\
&= S(\Delta_{01}(\Phi_{\mathcal{N}})) - S(\Delta_{1}(\Phi_{\mathcal{N}})) + D(\Delta_{01}(\Phi_{\mathcal{N}})||\Delta_{01}(\Phi_{\mathcal{M}}))\\
&\geq S(\Delta_{01}(\Phi_{\mathcal{N}})) - S(\Delta_{1}(\Phi_{\mathcal{N}})).
\end{split}
\end{equation}
In the last line, the state $\Phi_{\mathcal{M}} = \Delta_{01}(\mathcal{N})$ saturates the equality. This completes the proof.
% then the proof is completed.
\end{proof}
From the proof of Lemma~\ref{lemma:rel}, we can see that for any qcCRO, $\mathcal{M}$, the state $\Delta_{1}(\Phi_{\mathcal{M}})$ is an incoherent state on $\mathcal{H}_0\otimes \mathcal{H}_1$. Correspondingly, Eq.~\eqref{eq:RelofIr} is the relative entropy of coherence of $\Delta_{1}(\Phi_{\mathcal{N}})$. From Lemma~\ref{lemma:rel}, to obtain $C_{\mathrm{rel}}(\mathcal{N})$, one only needs to get $\Phi_{\mathcal{N}}$ and calculate the von Neumann entropy of $\Delta_{01}(\Phi_{\mathcal{N}})$ as well as $\Delta_{1}(\Phi_{\mathcal{N}})$. This calculation is essentially a matrix diagonalization problem, which can be solved in polynomial time.

In the discussion below, we take a subset of RNG superchannels to be the set of free superchannels $\mathcal{F}$,
\begin{equation}
	\mathcal{F} = \{\Lambda \in \mathrm{RNG} | \Delta\circ \Lambda(\mathcal{N}) = \Lambda(\Delta\circ\mathcal{N}), \forall \mathcal{N}\in\mathrm{CPTP} \}.
\end{equation}
The set $\mathcal{F}$ is not empty as the identity superchannel, $\Lambda(\mathcal{N}) = \mathcal{N}$, belongs to $\mathcal{F}$. It is also a convex set as $\Lambda_1,\Lambda_2\in\mathcal{F}\Rightarrow p\Lambda_1+(1-p)\Lambda_2\in\mathcal{F},\forall 0\leq p\leq 1$. Now we show that the relative entropy of irreplaceability has the properties of non-increasing under mixing and monotonicity under $\mathcal{F}$, as shown by the following lemma.

\begin{lemma}
The relative entropy of irreplaceability in Eq.~\eqref{eq:RelofIrDef}, $C_{\mathrm{rel}}(\mathcal{N})$, has the following properties:
\begin{enumerate}
\item
\emph{Non-increasing under mixing: }Given an index set $\mathcal{I}$, $\forall\mathcal{N}_{i\in\mathcal{I}}\in CPTP$, $\forall\{p_i\}_{i\in\mathcal{I}}$ such that $p_i\geq0,\sum_{i\in\mathcal{I}}p_i=1$,
\begin{equation}
C_{\mathrm{rel}}\left(\sum_{i}p_i\mathcal{N}_i\right)\leq \sum_{i}p_i C_{\mathrm{rel}}(\mathcal{N}_i).
\end{equation}

\item
\emph{Monotonicity under a free superchannel: }$\forall\mathcal{N}\in CPTP$, $\forall\Lambda\in \mathcal{F}$,
\begin{equation}
	C_{\mathrm{rel}}(\Lambda(\mathcal{N}))\leq C_{\mathrm{rel}}(\mathcal{N}).
\end{equation}
\end{enumerate}
\end{lemma}
\begin{proof}
\begin{equation}
	\begin{split}
		C_{\mathrm{rel}}\left(\sum_{i}p_i\mathcal{N}_i\right) &= \min_{\mathcal{M}\in \mathrm{qcCRO}} D\left(\sum_{i}p_i\Delta_{1}(\Phi_{\mathcal{N}_i})||\Delta_{1}(\Phi_{\mathcal{M}})\right)\\
		&\leq D\left(\sum_{i}p_i\Delta_{1}(\Phi_{\mathcal{N}_i})||\sum_{i}p_i\Delta_{1}(\Phi_{\mathcal{M}_i})\right)\\
		&\leq \sum_ip_i D(\Delta_{1}(\Phi_{\mathcal{N}_i})||\Delta_{1}(\Phi_{\mathcal{M}_i}))\\
		&= \sum_{i}p_i C_{\mathrm{rel}}(\mathcal{N}_i),
	\end{split}
\end{equation}
where $\mathcal{M}_i$ is the CPTP channel minimizing the value $D(\Delta_{1}(\Phi_{\mathcal{N}_i})||\Delta_{1}(\Phi_{\mathcal{M}}))$. The inequality of the third line comes from the joint convexity of relative entropy. For any free superchannel $\Lambda\in \mathcal{F}$,
\begin{equation}
	\begin{split}
		C_{\mathrm{rel}}(\Lambda(\mathcal{N})) &= \min_{\mathcal{M}\in\mathrm{qcCRO}} D(\Phi_{\Delta\circ \Lambda(\mathcal{N})}||\Phi_{\Delta\circ \mathcal{M}})\\
		&\leq \min_{\mathcal{M}\in\mathrm{qcCRO}} D(\Phi_{\Delta\circ \Lambda(\mathcal{N})}||\Phi_{\Delta\circ \Lambda(\mathcal{M})})\\
		&= \min_{\mathcal{M}\in\mathrm{qcCRO}} D(\Phi_{ \Lambda(\Delta\circ\mathcal{N})}||\Phi_{\Lambda(\Delta\circ\mathcal{M})})\\
		&\leq \min_{\mathcal{M}\in\mathrm{qcCRO}} D(\Phi_{\Delta\circ\mathcal{N}}||\Phi_{\Delta\circ \mathcal{M}})\\
		&= C_{\mathrm{rel}}(\mathcal{N}).
	\end{split}
\end{equation}
The inequality of the fourth line comes from the data processing inequality.
\end{proof}

Similar with the robustness of irreplaceability, the relative entropy of irreplaceability also enjoys the stability under extensions to large systems as shown in the following lemma.
\begin{lemma}\label{lemma:extensionrel}
The relative entropy of irreplaceability, $C_{\mathrm{rel}}(\mathcal{N})$, satisfies $C_{\mathrm{rel}}(\mathcal{N}) = C_{\mathrm{rel}}(\mathcal{N}\otimes I)$, where $I$ is the identity operation of an ancillary system with an arbitrary finite dimension.
\end{lemma}
\begin{proof}
Given a quantum channel, $\mathcal{N}$, on Hilbert space $\mathcal{H}_0$ and the identity operation $I_{0'}$ on an ancillary system $\mathcal{H}_{0'}$, the Choi-state of $\mathcal{N}$ and $\mathcal{N}\otimes I_{0'}$ are shown below.
\begin{align}
\Phi_{\mathcal{N}} &= I_1 \otimes \mathcal{N} (\ketbra{\Phi^+}_{10}),\\
\Phi_{\mathcal{N}\otimes I_{0'}} &= (I_1 \otimes \mathcal{N})\otimes (I_{1'}\otimes I_{0'}) (\ketbra{\Phi^+}_{10}\otimes \ketbra{\Phi^+}_{1'0'}) = \Phi_{\mathcal{N}}\otimes \ketbra{\Phi^+}_{1'0'}.
\end{align}
Here, $\ketbra{\Phi^+}_{10}$ and $\ketbra{\Phi^+}_{1'0'}$ are maximally entangled state on Hilbert space $\mathcal{H}_{1}\otimes \mathcal{H}_{0}$ and $\mathcal{H}_{1'}\otimes \mathcal{H}_{0'}$ respectively. $\mathcal{H}_{1}$ and $\mathcal{H}_{1'}$ are another two ancillary systems for defining Choi-states. According to Lemma~\ref{lemma:rel}, the relative entropy of irreplaceability of $\mathcal{N}\otimes I_{0'}$ is given by
\begin{equation}
\begin{split}
C_{\mathrm{rel}}(\mathcal{N}\otimes I_{0'}) &= S(\Delta_{00'11'}(\Phi_{\mathcal{N}\otimes I_{0'}})) - S(\Delta_{11'}(\Phi_{\mathcal{N}}\otimes I_{0'}))\\
&= S(\Delta_{01}(\Phi_{\mathcal{N}})\otimes \Delta_{0'1'}(\ketbra{\Phi^+}_{1'0'})) - S(\Delta_{1}(\Phi_{\mathcal{N}})\otimes \Delta_{1'}(\ketbra{\Phi^+}_{1'0'}))\\
&= ( S(\Delta_{01}(\Phi_{\mathcal{N}})) - S(\Delta_{1}(\Phi_{\mathcal{N}})) ) + ( S(\Delta_{0'1'}(\ketbra{\Phi^+}_{1'0'})) - S(\Delta_{1'}(\ketbra{\Phi^+}_{1'0'})) )\\
&= C_{\mathrm{rel}}(\mathcal{N}) + ( S(\Delta_{0'1'}(\ketbra{\Phi^+}_{1'0'})) - S(\Delta_{0'1'}(\ketbra{\Phi^+}_{1'0'})) )\\
&= C_{\mathrm{rel}}(\mathcal{N}).
\end{split}
\end{equation}
Here, the line 3 utilises $S(\rho\otimes \sigma) = S(\rho) + S(\sigma)$ and the line 4 uses the equality $\Delta_{1'}(\ketbra{\Phi^+}_{1'0'}) = \Delta_{0'1'}(\ketbra{\Phi^+}_{1'0'})$. Now we obtain $C_{\mathrm{rel}}(\mathcal{N}\otimes I_{0'}) = C_{\mathrm{rel}}(\mathcal{N})$ and complete the proof.
\end{proof}

\end{document}